\documentclass{article}
\ifdefined\pdfinfoomitdate\pdfinfoomitdate=1\fi
\ifdefined\pdftrailerid\pdftrailerid{}\fi
\ifdefined\pdfsuppressptexinfo\pdfsuppressptexinfo=-1\fi
\usepackage[T1]{fontenc}
\usepackage{microtype}
\usepackage{graphicx}
\usepackage{booktabs}
\usepackage{tabularx}
\usepackage{amsmath,amssymb,amsthm}
\usepackage{xcolor}
\usepackage{hyperref}

\usepackage[accepted]{mlsys2025}
\usepackage{float}
\makeatletter
\renewcommand{\mlsys@appearing}{Preprint.}
\renewcommand{\printAffiliationsAndNotice}[1]{%
\stepcounter{@affiliationcounter}%
{\let\thefootnote\relax\footnotetext{\raggedright\hspace*{-\footnotesep}\ifdefined\isaccepted #1\fi%
\forloop{@affilnum}{1}{\value{@affilnum} < \value{@affiliationcounter}}{
\textsuperscript{\arabic{@affilnum}}\ifcsname @affilname\the@affilnum\endcsname%
\csname @affilname\the@affilnum\endcsname%
\else
{\bf AUTHORERR: Missing \textbackslash{}mlsysaffiliation.}
\fi
}.
\ifdefined\mlsyscorrespondingauthor@text
Correspondence to: \mlsyscorrespondingauthor@text.
\else
{\bf AUTHORERR: Missing \textbackslash{}mlsyscorrespondingauthor.}
\fi

\ \\
\Notice@String
}
}
}
\makeatother
\usepackage[capitalise]{cleveref}

\newtheorem{proposition}{Proposition}
\theoremstyle{remark}
\newtheorem{remark}{Remark}

\newcommand{\BoneGsmMed}{0.91}

\newcommand{\BtwoNfp}{40}
\newcommand{\BtwoDeepBf}{6.7}

\newcommand{\BtwoDeepFp}{15.0}
\newcommand{\BtwoDeepFpLo}{14.2}
\newcommand{\BtwoDeepFpHi}{15.8}

\newcommand{\BtwoSfourFirst}{0.990}

\newcommand{\BtwoShareEightFirst}{0.980}

\newcommand{\BtwoShareFourFirst}{0.970}

\newcommand{\BtwoWinFourFirst}{0.845}

\newcommand{\BthreeFullFour}{3.31}
\newcommand{\BthreeSfourFour}{4.33}

\newcommand{\BthreeN}{60}

\newcommand{\BfourRerunSame}{30\%}
\newcommand{\BfourRerunKL}{1.6\times10^{-4}}

\newcommand{\BfiveSpeedZero}{2.64}
\newcommand{\BfiveAccZero}{0.878}

\newcommand{\BfiveSameZero}{31}
\newcommand{\BfiveSpeedOne}{4.14}
\newcommand{\BfiveAccOne}{0.965}

\newcommand{\BsixRatio}{2.24}
\newcommand{\BsevenSamePos}{100\%}

\newcommand{\BsevenConeLabel}{exceeds a precision change}

\newcommand{\BsevenConePeak}{13.8}
\newcommand{\BsevenConePeakLo}{2.28}
\newcommand{\BsevenConePeakHi}{21.2}
\newcommand{\BsevenPeakShare}{75\%}

\newcommand{\BsevenMedSfour}{6.4\times10^{-5}}
\newcommand{\BsevenMedFp}{8.2\times10^{-5}}
\newcommand{\BsevenMeanRatio}{5.7}
\newcommand{\BsevenBinMax}{3.3}
\newcommand{\BsevenBinMin}{1.6}
\newcommand{\BsevenMaxSfour}{5.05}

\newcommand{\BsevenShareSfour}{74\%}
\newcommand{\BsevenShareFp}{21\%}

\newcommand{\BsevenWsameZero}{31}
\newcommand{\BsevenWsameOne}{31}
\newcommand{\BnineDeepBf}{3.8}

\newcommand{\BnineDeepFp}{12.8}
\newcommand{\BnineDeepFpLo}{12.0}
\newcommand{\BnineDeepFpHi}{13.7}

\newcommand{\BnineNfp}{40}
\newcommand{\BnineSsixFirst}{0.98}
\newcommand{\BnineConeLabel}{indistinguishable from a precision change}

\newcommand{\BnineBinMin}{1.04}
\newcommand{\BnineBinMax}{1.51}
\newcommand{\BnineMaxSsix}{0.020}
\newcommand{\BnineMaxFp}{0.014}

\newcommand{\BnineFlipSsix}{48}
\newcommand{\BnineFlipFp}{76}

\newcommand{\BeightAPrompts}{32}
\newcommand{\BeightBvfSame}{20}

\newcommand{\BeightAoneRatio}{1.37}
\newcommand{\BeightAoneLo}{1.20}
\newcommand{\BeightAoneHi}{1.62}
\newcommand{\BeightAoneLabel}{indistinguishable from numerics}
\newcommand{\BeightAtwoRatio}{0.94}

\newcommand{\BeightBvSfourSame}{28}

\newcommand{\BeightTpvFull}{3.34}
\newcommand{\BeightTpvSfour}{4.47}

\newcommand{\BeightSeqRatio}{1.00}
\newcommand{\BeightMetricsGsm}{24}

\newcommand{\BeightTPrompts}{12}
\newcommand{\BeightSpeedFull}{2.09}
\newcommand{\BeightSpeedSfour}{2.58}

\newcommand{\BeightTailC}{9}
\newcommand{\BeightFlipC}{35}
\newcommand{\BeightMaxC}{5.05}

\newcommand{\BeightTailWsixteen}{5}
\newcommand{\BeightFlipWsixteen}{37}
\newcommand{\BeightMaxWsixteen}{0.97}

\newcommand{\BeightTailWsixtyfour}{1}
\newcommand{\BeightFlipWsixtyfour}{31}
\newcommand{\BeightMaxWsixtyfour}{0.13}

\newcommand{\BeightTailWtwofivesix}{0}
\newcommand{\BeightFlipWtwofivesix}{10}
\newcommand{\BeightMaxWtwofivesix}{0.04}

\newcommand{\BeightTailKfour}{9}
\newcommand{\BeightFlipKfour}{38}
\newcommand{\BeightMaxKfour}{5.03}

\newcommand{\BeightTailSeight}{9}
\newcommand{\BeightFlipSeight}{42}
\newcommand{\BeightMaxSeight}{5.24}

\newcommand{\BeightTailF}{1}
\newcommand{\BeightFlipF}{32}
\newcommand{\BeightMaxF}{0.23}
\newcommand{\BeightBoneRatio}{0.11}
\newcommand{\BeightBoneLo}{0.00}
\newcommand{\BeightBoneHi}{0.20}
\newcommand{\BeightBoneRed}{89\%}

\newcommand{\BeightBPrompts}{48}

\newcommand{\BeightConeRatio}{72}

\newcommand{\BeightCtwoRatio}{1.88}
\newcommand{\BeightCtwoLo}{1.65}
\newcommand{\BeightCtwoHi}{2.20}

\newcommand{\BeightCthreeRatio}{0.88}

\newcommand{\BeightCfourRatio}{0.33}

\newcommand{\BeightDratio}{2.26}
\newcommand{\MHugFullKB}{2{,}788}
\newcommand{\MHugDeepKB}{169}
\newcommand{\MHugDeepFrac}{6.1\%}
\newcommand{\MHugTwoFrac}{9.1\%}
\newcommand{\MHugFullMB}{2{,}855}
\newcommand{\MHugDeepMB}{173}
\newcommand{\MHugRipeSixteenMB}{215}
\newcommand{\MHugRipeSixteenFrac}{7.5\%}
\newcommand{\MHugRipeSixteenExtra}{42}
\newcommand{\MHugRipeSixtyFourMB}{341}
\newcommand{\MHugRipeSixtyFourFrac}{11.9\%}
\newcommand{\MHugRipeSixtyFourExtra}{168}
\newcommand{\MHugRipeTwoFiveSixMB}{843}
\newcommand{\MHugRipeTwoFiveSixFrac}{29.5\%}
\newcommand{\MHugRipeTwoFiveSixExtra}{670}
\newcommand{\MRavFullKB}{410}
\newcommand{\MRavDeepKB}{28.7}
\newcommand{\MRavDeepFrac}{7.0\%}
\newcommand{\MRavTwoFrac}{10.0\%}
\newcommand{\MRavFullMB}{419}
\newcommand{\MRavDeepMB}{29}
\newcommand{\MRavRipeSixteenMB}{35}
\newcommand{\MRavRipeSixteenFrac}{8.5\%}

\newcommand{\MRavRipeSixtyFourMB}{54}
\newcommand{\MRavRipeSixtyFourFrac}{12.8\%}

\newcommand{\MRavRipeTwoFiveSixMB}{127}
\newcommand{\MRavRipeTwoFiveSixFrac}{30.2\%}

\newcommand{\MHugSinkExtra}{10}
\newcommand{\MHugTwoExtra}{87}
\newcommand{\MCells}{26}
\newcommand{\MCacheErr}{0.005}
\newcommand{\MMaxSeq}{808}
\newcommand{\MOzero}{7.2}
\newcommand{\MOoneFull}{0.18}
\newcommand{\MOoneShared}{0.15}
\newcommand{\MPeakOosErr}{0.15\%}
\newcommand{\MPeakOosN}{9}
\newcommand{\MPerSeqFull}{2.25}

\newcommand{\MBmaxFull}{36}
\newcommand{\MBmaxShared}{302}

\newcommand{\MTLossSixInf}{449}

\newcommand{\MTLossSixRes}{2\%}
\newcommand{\MTShSixInf}{689}

\newcommand{\MTShSixRes}{6\%}

\newcommand{\MGHugFp}{2.70}
\newcommand{\MGHugFpLo}{2.52}
\newcommand{\MGHugFpHi}{2.87}
\newcommand{\MSlopeShHugFp}{2.06}
\newcommand{\MSlopeDpHugFp}{2.71}

\newcommand{\MGRavFp}{1.82}
\newcommand{\MGRavFpLo}{1.68}
\newcommand{\MGRavFpHi}{1.97}

\newcommand{\MCensHugBf}{23\%}

\newcommand{\MCensRavBf}{21\%}
\newcommand{\MRhoHugK}{0.88}
\newcommand{\MHalfHugK}{5.2}
\newcommand{\MFloorHugK}{0.0072}
\newcommand{\MDeepBfHugK}{0.0116}
\newcommand{\MDeepFpHugK}{0.0097}
\newcommand{\MResFpHugK}{0.28}
\newcommand{\MResBfHugK}{0.27}
\newcommand{\MRhoHugV}{0.87}
\newcommand{\MHalfHugV}{5.1}
\newcommand{\MFloorHugV}{0.0115}
\newcommand{\MDeepBfHugV}{0.0192}
\newcommand{\MDeepFpHugV}{0.0162}

\newcommand{\MRhoRavK}{0.81}
\newcommand{\MHalfRavK}{3.3}
\newcommand{\MFloorRavK}{0.0048}
\newcommand{\MDeepBfRavK}{0.0047}
\newcommand{\MDeepFpRavK}{1.6\times10^{-4}}
\newcommand{\MResFpRavK}{0.51}
\newcommand{\MResBfRavK}{0.19}
\newcommand{\MRhoRavV}{0.81}
\newcommand{\MHalfRavV}{3.3}
\newcommand{\MFloorRavV}{0.0101}
\newcommand{\MDeepBfRavV}{0.0098}
\newcommand{\MDeepFpRavV}{3.7\times10^{-4}}

\newcommand{\MDeepRatioFpK}{61}
\newcommand{\MErrPredHugFp}{3.6\times10^{-6}}

\newcommand{\MErrSHugFp}{3.2\times10^{-6}}

\newcommand{\MErrShallowHugFp}{2.4\times10^{-9}}

\newcommand{\MErrPredRavFp}{1.8\times10^{-8}}

\newcommand{\MErrSmeanRavFp}{2.7\times10^{-8}}

\newcommand{\MErrRatio}{202}
\newcommand{\MTailN}{9}
\newcommand{\MTailPrompts}{4}
\newcommand{\MTailTop}{5}
\newcommand{\MTailLowMargin}{2}
\newcommand{\MLowMarginFrac}{8\%}
\newcommand{\MTailPositions}{7{,}130}
\newcommand{\MRavPositions}{7{,}647}

\newcommand{\MTailFpHug}{1}

\newcommand{\MTraceWsixteen}{0.97}
\newcommand{\MTraceWsixtyfour}{0.13}

\newcommand{\MAlphaFullFour}{0.81}

\newcommand{\MSpreadFull}{0.06}
\newcommand{\MRejRhoFull}{0.76}

\newcommand{\MAlphaDeepFour}{0.93}

\newcommand{\MSpreadDeep}{0.04}
\newcommand{\MRejRhoDeep}{0.66}
\newcommand{\MTauErrMax}{2.6\%}
\newcommand{\MTauOosErr}{6\%}
\newcommand{\MTauPredFullOos}{3.40}
\newcommand{\MTauPredDeepOos}{4.33}
\newcommand{\MTauTimFull}{3.61}
\newcommand{\MTauTimDeep}{4.53}
\newcommand{\MEta}{1.07}

\newcommand{\MCfull}{132}
\newcommand{\MSpZeroFull}{2.24}
\newcommand{\MSpOneFull}{2.08}
\newcommand{\MSpMeasFull}{2.09}
\newcommand{\MSpZeroDeep}{2.80}
\newcommand{\MSpOneDeep}{2.61}
\newcommand{\MSpMeasDeep}{2.58}
\newcommand{\MSpPromptErr}{5.4\%}
\newcommand{\MSpPrompts}{12}

\newcommand{\MHugTwoMB}{260}
\newcommand{\MRavTwoMB}{42}
\newcommand{\MSpTimings}{24}

\newcommand{\MDecEq}{1.67}
\newcommand{\MDecBatch}{1.34}
\newcommand{\MDecTotal}{2.24}

\newcommand{\MDecBeq}{32}
\newcommand{\MDecSharedB}{64}

\newcommand{\MErrSizeHugFp}{1.7\times10^{-6}}
\newcommand{\MErrSizeRavFp}{4.7\times10^{-9}}
\newcommand{\MErrRatioSize}{366}
\newcommand{\MTailLowExp}{0.7}
\newcommand{\MTailLowOver}{2.7}
\newcommand{\MDesignDeepFourGamma}{8}
\newcommand{\MDesignDeepFourS}{2.87}
\newcommand{\MDesignFullFourGamma}{5}
\newcommand{\MDesignFullFourS}{1.98}

\newcommand{\BrOneRatio}{1.24}
\newcommand{\BrOneLo}{0.99}
\newcommand{\BrOneHi}{1.57}
\newcommand{\BrOneLabel}{indistinguishable from numerics}
\newcommand{\BrPrompts}{22}
\newcommand{\BrPlanned}{32}
\newcommand{\BrDdEqSeq}{19}

\newcommand{\BrTauFull}{3.39}
\newcommand{\BrSpeedFull}{1.87}
\newcommand{\BrTauSfour}{4.46}
\newcommand{\BrSpeedSfour}{2.77}
\newcommand{\BrTauRipe}{3.42}
\newcommand{\BrSpeedRipe}{1.81}
\newcommand{\BrTauDd}{4.32}
\newcommand{\BrSpeedDd}{2.33}
\newcommand{\BrRecDd}{88\%}
\newcommand{\BrRecDdLo}{79\%}
\newcommand{\BrRecDdHi}{97\%}
\newcommand{\BrRecRipe}{3\%}
\newcommand{\BrRecDdLabel}{keeps the shared gain}
\newcommand{\BrThreeRatio}{0.80}
\newcommand{\BrThreeLo}{0.39}
\newcommand{\BrThreeHi}{1.25}
\newcommand{\BrThreeLabel}{no difference}
\newcommand{\BrTailDd}{1}
\newcommand{\BrTailSfour}{3}
\newcommand{\BrMaxDd}{0.94}
\newcommand{\BrMaxSfour}{0.53}

\newcommand{\BrPredRipe}{2.06}
\newcommand{\BrMeasRipe}{2.06}
\newcommand{\BrErrRipe}{$-$0.1\%}
\newcommand{\BrPredDd}{2.59}
\newcommand{\BrMeasDd}{2.66}
\newcommand{\BrErrDd}{+2.5\%}

\newcommand{\BrAncPredFull}{2.07}
\newcommand{\BrAncMeasFull}{1.87}
\newcommand{\BrAncErrFull}{$-$10\%}
\newcommand{\BrAncPredSfour}{2.61}
\newcommand{\BrAncMeasSfour}{2.77}
\newcommand{\BrAncErrSfour}{+6\%}
\newcommand{\BrAncErrMax}{10\%}
\newcommand{\BrSeqRatio}{0.88}
\newcommand{\BrTimingPrompts}{12}

\newcommand{\MRejKFull}{0.62}
\newcommand{\MRejKDeep}{0.42}

\newcommand{\ReplayPositions}{4{,}541}
\newcommand{\ReplayBootstrapDraws}{2{,}000}
\newcommand{\ReplayBootstrapUndefined}{28}
\newcommand{\ReplayBootstrapDefined}{1{,}972}
\newcommand{\ReplayLooMin}{85.7\%}
\newcommand{\ReplayLooMax}{100.0\%}
\newcommand{\ReplayTailPromptsC}{4}
\newcommand{\ReplayTailPromptsWsixteen}{3}
\newcommand{\ReplayTailPromptsWsixtyfour}{1}
\newcommand{\ReplayTailPromptsWtwofivesix}{0}
\newcommand{\ReplayTailPromptsKfour}{4}

\newcommand{\ReplayTailPromptsSeight}{4}
\newcommand{\ReplayTailPromptsF}{1}

\newcommand{\sysname}{RIPE}
\newcommand{\sysexpansion}{Recent-window Intact, deepest Plane Elsewhere}
\newcommand{\MOuroPromptShared}{80\%}
\newcommand{\KL}{\mathrm{KL}}
\newcommand{\npre}{n_{\mathrm{pre}}}
\newcommand{\ncore}{n_{\mathrm{core}}}
\newcommand{\ncoda}{n_{\mathrm{coda}}}

\definecolor{hug}{HTML}{1F5FA8}
\definecolor{rav}{HTML}{D9730D}
\definecolor{deepc}{HTML}{1F5FA8}
\definecolor{shallowc}{HTML}{8A8A8A}
\definecolor{ripe}{HTML}{2E8B57}
\definecolor{w16}{HTML}{98C998}
\definecolor{w64}{HTML}{2E8B57}
\definecolor{w256}{HTML}{0B4F2B}

\mlsystitlerunning{Recency Is the Wrong Axis for Looped KV Caches}

\begin{document}

\twocolumn[
\mlsystitle{Recency Is the Wrong Axis: A Negative Result on\\Windowed Deepest-Plane KV Caching in Looped Language Models}

\begin{mlsysauthorlist}
\mlsysauthor{Dezhen Lu}{aff1}
\mlsysauthor{Runguo Li}{aff2}
\mlsysauthor{Zixiang Ni}{aff3}
\mlsysauthor{Yuchen Tang}{zstu}
\mlsysauthor{Xingquan Kong}{swpu}
\mlsysauthor{Jiaxing Yang}{fudan}
\mlsysauthor{Jiarui Guo}{nankai}
\end{mlsysauthorlist}

\mlsysaffiliation{aff1}{University of T\"ubingen, T\"ubingen, Germany}
\mlsysaffiliation{aff2}{University of Illinois Urbana-Champaign, Urbana, IL, USA}
\mlsysaffiliation{aff3}{Xi'an Jiaotong University, Xi'an, China}
\mlsysaffiliation{zstu}{Zhejiang Sci-Tech University, China}
\mlsysaffiliation{swpu}{Southwest Petroleum University, China}
\mlsysaffiliation{fudan}{Fudan University, Shanghai, China}
\mlsysaffiliation{nankai}{Nankai University, Tianjin, China}
\mlsyscorrespondingauthor{Dezhen Lu}{dezhen.lu@student.uni-tuebingen.de}
\mlsyscorrespondingauthor{Runguo Li}{runguol2@illinois.edu}
\mlsyscorrespondingauthor{Zixiang Ni}{nzx@stu.xjtu.edu.cn}

\mlsyskeywords{looped transformers, recurrent depth, KV cache, speculative decoding, performance modeling}

\vskip 0.3in

\begin{abstract}
Looped language models reuse one block of layers over $R$ recurrent steps but store a key--value (KV) plane for each step.
Keeping only the deepest plane shrinks the cache up to $R$-fold; a natural repair, which we call \sysname{}, keeps every plane
for the $W$ most recent positions and shares the deepest plane elsewhere. This repair acts on the wrong axis: across three
looped models, the cost of sharing is set by how far the late loops' planes have converged (\emph{ripeness}) and by which
positions share, not by recency. On Raven-Llama the deep planes sit within the bf16 rounding floor, so no window is needed.
On Huginn-0125 a 64-position window removes most of a concentrated replay-error tail (\BeightTailC{} to
\BeightTailWsixtyfour{} positions over \BeightBPrompts{} prompts), yet in a real self-speculative decoder the policy is slower
and larger than the deepest plane alone and not measurably closer to full-cache decoding. On Ouro-1.4B the prompt lies
outside any 16--64-position window: on 24 GSM8K questions with a 2048-token budget, the full cache and a FlashLoop baseline
finish 23 answers and \sysname{} at most 2, although physical storage matches the emulated policy bitwise. Sensitivity and
ripening models fitted on Huginn and Raven explain their contrast, and Ouro's own report that prompt-side sharing is harmful explains the
third. Separately, bf16 execution-shape effects are as large as the cache effects being measured. Recent-window repair of
shared-plane caches should not be adopted without a per-model ripeness test; we release ten pre-registered batches with
every failed prediction.
\end{abstract}
]

\printAffiliationsAndNotice{}
\hypersetup{pdftitle={Recency Is the Wrong Axis: A Negative Result on Windowed Deepest-Plane KV Caching in Looped Language Models},
pdfsubject={MLSys 2027 submission}}

\section{Introduction}

\begin{figure*}[t]
\centering
\includegraphics{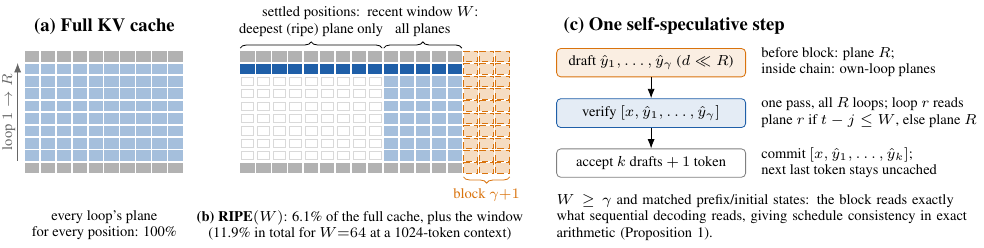}
\caption{Overview of \sysname{} (schematic; grid sizes not to scale). Rows are the prelude (bottom), loops $1,\dots,R$ and
the coda (top) of each prefix position (columns). (a)~A full cache keeps every loop's keys and values. (b)~\sysname{} keeps
only the deepest plane for settled positions and all planes for the $W$ most recent ones; the block being verified
lies inside the window. Memory from Eqs.~\eqref{eq:mem} and~\eqref{eq:ripemem}. (c)~A self-speculative step with
$W=64$, implemented under emulated storage in Section~\ref{sec:ripe-e2e}.}\label{fig:overview}
\end{figure*}

Looped, or recurrent-depth, language models reuse one block of transformer layers many times per token
\citep{huginn,ouro,retrofit}. Because a shallow run of the same weights can serve as a draft, they support
self-speculative decoding without a separate draft model by drafting with $d \ll R$ loops and verifying at full depth
\citep{huginn,loopspec,wavefront,maturespec,loopahead}. Weight reuse keeps parameter storage small, but the KV cache grows
with recurrent depth. Every loop of every recurrent layer writes its own keys and values, so Huginn-0125 stores \MHugFullKB{}\,KB per token, against
\MHugDeepKB{}\,KB if only one plane were kept. At the \MMaxSeq-token sequence length of the batch-6 and batch-8 memory cells, a 95\,GB budget then holds \MBmaxFull{} sequences
with the full cache and \MBmaxShared{} with one plane (Section~\ref{sec:memory}).

Which planes have to be stored? Huginn's own code offers caches that keep fewer planes \citep{huginn}, Ouro reports that
the last loop's KV suffices for its model's decoded tokens, though not for the prompt \citep{ouro}, and WaveFront shows that cross-loop sharing raises
self-speculative speed \citep{wavefront}. Keeping only the deepest plane is the cheapest choice, and its errors invite an
obvious repair: keep every plane for the $W$ most recent positions, where decoding and speculative verification happen, and
let older positions share the deepest plane (Figure~\ref{fig:overview}). The window also makes batched verification
consistent with sequential decoding. We call this policy \sysname{} (\emph{\sysexpansion}), built it, and tested it on three
looped models in ten pre-registered GPU batches, whose decision rules and predictions were frozen before each run, and in
audited reruns on other hardware. Many predictions failed; the supplementary appendix lists all of them.

This paper is the resulting negative result. The window selects positions by recency, but on every model the cost of the
deepest plane is set by how far the late loops' planes have converged, which we call \emph{ripeness} and which is a property of
the model, and by which positions share, which for a prompt of a few hundred tokens is almost all of them. Small post hoc
models (Table~\ref{tab:models}) make this quantitative on Huginn and Raven.

\textbf{Contributions.}
\vspace{-0.6\baselineskip}
\begin{itemize}\setlength{\itemsep}{1pt}\setlength{\parskip}{0pt}
\item \emph{A three-part negative result.} On Raven-Llama the deepest plane is already indistinguishable from bf16
rounding, so the window is unnecessary (Section~\ref{sec:fidelity}). On Huginn-0125 a 64-position window removes most
of a replay-error tail, but the \sysname{} decoder is slower, larger and not measurably closer to the full cache than the
deepest plane alone, so the window is ineffective (Section~\ref{sec:ripe-e2e}). On pretrained Ouro-1.4B the prompt lies
outside the window, and with $W\le64$ the model almost never finishes an answer that the full cache finishes, so the window
is insufficient (Section~\ref{sec:ouro}).
\item \emph{Why: ripeness and position, not recency.} A sensitivity model (KL grows as a power \MGRavFp--\MGHugFp{} of a
prefix-KV perturbation, late loops far more sensitive) and a ripening model (approximately geometric convergence of planes
with a bf16 floor) predict the deepest plane's error to within a factor of two and separate Raven from Huginn
(Sections~\ref{sec:sens},~\ref{sec:ripen}). Neither quantity depends on a position's recency; the window changes how many
positions share, not what sharing costs them.
\item \emph{Replay tails are not decoder quality.} The teacher-forced replay tail that motivated the window falls by
\BeightBoneRed{} on Huginn, while the decoder it motivated shows no detectable difference on 22 problems
(ratio \BrThreeRatio{} [\BrThreeLo, \BrThreeHi]). Per-head position filtering, attention sinks and a second plane do not
help either (Section~\ref{sec:negative}).
\item \emph{What the window does provide, and what it does not.} When $W\ge\gamma$, batched verification under
\sysname{} computes the same values as sequential \sysname{} decoding in exact arithmetic (Section~\ref{sec:ripe}); deepest-plane
drafts keep \BrRecDd{} of the deepest plane's acceptance gain. These concern schedules and speed, not closeness to the full cache. In bf16, execution shape alone changes
trajectories at rates comparable to the cache effects: cuDNN attention makes GH200 decoding non-reproducible, and Ouro's
sequential-versus-batched gate fails in bf16 on two GPUs but passes in fp32 (Sections~\ref{sec:repro},~\ref{sec:ouro}).
\item \emph{Performance models and audited evidence.} Acceptance, cost, memory and throughput models, two checked on
held-out data (Sections~\ref{sec:accept}--\ref{sec:memory}); frozen inputs and independent audits of every rerun.
\end{itemize}

\section{Background}\label{sec:background}

\textbf{Looped models.} A prelude $P$ of $\npre$ layers embeds token $x_t$; a core block $B$ of $\ncore$ layers runs $R$
times on a state initialized at random; a coda $C$ of $\ncoda$ layers produces logits:
\begin{align}
s^{(r)}_t &= B\big(s^{(r-1)}_t,\, P(x_t);\ \mathcal{C}^{(r)}_{<t}\big),\qquad r=1,\dots,R,\\
z_t &= C\big(s^{(R)}_t\big).
\end{align}
At loop $r$, each core layer emits keys and values for token $t$; we call loop $r$'s keys and values over all core layers
\emph{plane} $r$, written $(K^{(r)}, V^{(r)})$. In the full cache, the query of token $t$ at loop $r$ attends to plane $r$
of every prefix position. Huginn-0125 \citep{huginn} has $\npre=\ncoda=2$, $\ncore=4$, $R=32$, and 55 heads of width 96.
Raven-Llama \citep{retrofit} has $\npre=\ncoda=4$, $\ncore=6$, $R=32$, and grouped-query attention with 8 KV heads of
width 64.

\textbf{Cache policies.} A policy $\pi(r,t,j)\in\{1,\dots,R\}$ names the plane of prefix position $j$ that token $t$ reads
at loop $r$ (Figure~\ref{fig:overview}). The full cache has $\pi=r$; the \emph{deepest-plane} cache has $\pi=R$. Huginn's
\texttt{compress-s$S$} keeps $S$ core slots and overwrites them modulo $S$; with $S=\ncore$ the last writer of every slot
is loop $R$, so it is exactly the deepest-plane cache. Huginn's slot map hard-codes $\npre=2$; for Raven we use the
general map $\text{slot}=(\text{step}-\npre)\bmod S+\npre$, checked on CPU against direct substitution. With
$b=2\cdot2\,H_{\mathrm{kv}}\,d_h$ bytes per layer per token in bf16 (K and V, two bytes each), an $L$-token sequence needs
\begin{align}
M_{\mathrm{full}}(L)&=bL\,(\npre+\ncoda+R\,\ncore),\label{eq:mem}\\
M_{\mathrm{deep}}(L)&=bL\,(\npre+\ncoda+\ncore),\notag
\end{align}
that is, \MHugFullKB{} vs.\ \MHugDeepKB{}\,KB per token on Huginn (\MHugDeepFrac) and \MRavFullKB{} vs.\ \MRavDeepKB{}\,KB
on Raven (\MRavDeepFrac).

\textbf{Self-speculative decoding.} A draft runs $\gamma$ tokens through $d$ loops each; one verification pass runs the
block of $\gamma+1$ tokens through all $R$ loops at once, accepts the longest prefix on which draft and verifier agree
(greedy), and adds one token from the verifier. We write $\tau$ for tokens per verification.

\textbf{Measurement modes.} \emph{Teacher-forced prefill} replaces attention with an explicit fp32 softmax, records each
call's keys and values, and substitutes any loop's prefix K/V by another plane; the last 64 positions act as decoded
tokens. \emph{Generation} decodes greedily with the full cache, then replays each trajectory through the arm and through
an fp32 full cache, which sets the precision floor. We report $\KL(\text{reference}\,\|\,\text{arm})$ per position.

\begin{table*}[t]\centering\small
\caption{Post hoc models of this paper and how each was checked. ``Fitted on'' names the data that set its parameters;
``Check'' says whether the comparison data were held out from the fit. Batch numbers refer to the supplementary appendix.}\label{tab:models}
\begin{tabularx}{\textwidth}{@{}l>{\raggedright}p{3.3cm}>{\raggedright}X>{\raggedright}X>{\raggedright}p{1.5cm}>{\raggedright\arraybackslash}p{3.1cm}@{}}\toprule
Model & Form & Fitted on & Checked against & Check & Result\tabularnewline\midrule
KV memory & Eq.~\eqref{eq:mem}, \eqref{eq:ripemem} & architecture only & \MCells{} WaveFront cells; window extras (b.~8) & held out & $\le\MCacheErr$\,GB (print precision); exact\tabularnewline
Peak memory & cache $+\,o_0+o_1B$ & batch-6 cells & batch-8 cells, other node, $B\le256$ & held out & $\le\MPeakOosErr$\tabularnewline
Sensitivity & $\log\KL=a+b[\text{deep}]+g\log\|\delta\|$ & region arms (pre-registered) & second-order theory, $g=2$ & theory & $g=\MGHugFp$ / \MGRavFp{} (fp32)\tabularnewline
Ripening & $\Delta_r=c\rho^r$; bf16 adds $\Delta_\infty$ & fp32 curve, loops 8--31 & fp32 curve; bf16 curve with one more parameter & in sample & max $|\log$ err$|$ fp32 \MResFpHugK{} / \MResFpRavK, bf16 \MResBfHugK{} / \MResBfRavK\tabularnewline
Deepest-plane error & Eq.~\eqref{eq:err} & sensitivity + ripening & deepest-plane cache arm (fp32), not in the fit & partly held out & \(\MErrPredHugFp\) vs.\ \(\MErrSHugFp\) (factor 2)\tabularnewline
Acceptance & $\tau=(1-\alpha^{\gamma+1})/(1-\alpha)$ & batch-3 simulation & batch-8 real decoder & held out & within \MTauOosErr\tabularnewline
Speed-up & Eq.~\eqref{eq:speed} & one factor $\eta$, \MSpTimings{} timings (batch 8) & batch R, another machine: \sysname{} arms (pre-registered), anchors (post hoc) & held out & \BrErrRipe, \BrErrDd; anchors \BrAncErrFull, \BrAncErrSfour\tabularnewline\bottomrule
\end{tabularx}
\end{table*}

\section{Models}\label{sec:models}

All models in this section are post hoc: we built them after the pre-registered batches had run. We label each check by
whether the data it uses entered the fit.

\subsection{Sensitivity: deep loops are more sensitive to prefix KV}\label{sec:sens}

Let $\delta_r$ be the relative perturbation of the prefix keys and values that loop $r$ reads. To first order the logits move
by $\Delta z\approx\sum_r J_r\delta_r$. With $p=\mathrm{softmax}(z)$ and $p'=\mathrm{softmax}(z+\Delta z)$,
\begin{equation}\label{eq:kl}
\KL(p\,\|\,p')=\tfrac12\,\Delta z^{\top}F(p)\,\Delta z+O(\|\Delta z\|^3),
\end{equation}
where $F(p)=\mathrm{diag}(p)-pp^{\top}$. Dropping cross terms, $\KL\approx\sum_r s_r\|\delta_r\|^2$ for small perturbations:
quadratic, weighted by a per-loop sensitivity $s_r$. We fit the exponent rather than fix it at 2.

The pre-registered test of batches 2 and 9 substitutes the prefix K/V of either loops 2--9 (shallow) or loops 25--32 (deep)
by plane $p\in\{1,4,8,16,24,32\}$, measures the size $\|\delta\|$ as the root-sum-square of the relative distances over
loops, layers, keys and values, and regresses $\log\KL$ at the first decoded position on region and $\log\|\delta\|$
(Figure~\ref{fig:mech}a,b). In fp32 the size exponent is \MGRavFp{} [\MGRavFpLo, \MGRavFpHi] on Raven, consistent with the second-order
term, and $g=\MGHugFp$ [\MGHugFpLo, \MGHugFpHi] on Huginn, which excludes 2: Huginn's shallow loops follow the quadratic
(slope \MSlopeShHugFp) but its deep loops are steeper (\MSlopeDpHugFp). The largest substitutions (plane 1) are outside the
small-perturbation regime of Eq.~\eqref{eq:kl}; Raven's shallow region sits at the fp32 resolution. At equal size the deep region is $e^{\BtwoDeepFp}$
[$e^{\BtwoDeepFpLo}$, $e^{\BtwoDeepFpHi}$] (Huginn, \BtwoNfp{} texts) and $e^{\BnineDeepFp}$ [$e^{\BnineDeepFpLo}$,
$e^{\BnineDeepFpHi}$] (Raven, \BnineNfp{} texts) times more sensitive. In bf16 the coefficients are \BtwoDeepBf{} and \BnineDeepBf;
there \MCensHugBf{} and \MCensRavBf{} of the rows sit at the bf16 KL floor, which caps the measurable gap.

The consequence for caching: accuracy of the deep loops' prefix K/V deserves priority. At equal storage, keeping the
deepest plane agrees with the full cache at the first decoded position in \BtwoSfourFirst{} of texts, whereas a window of the
first four loops, the starting plan's choice, reaches \BtwoWinFourFirst{} (batch 2; sharing every 8 or 4 loops gives
\BtwoShareEightFirst{} and \BtwoShareFourFirst).

\begin{figure*}[t]
\centering
\includegraphics[width=\textwidth]{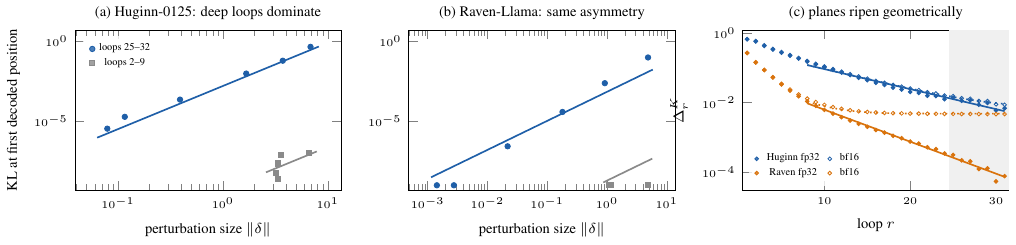}
\caption{Greater sensitivity in late loops, and when the deepest plane is enough. (a,b)~First-decoded-position KL after
substituting the prefix keys and values of loops 2--9 or 25--32 by planes 1--32 (fp32, AlpacaEval; points: arm medians over
texts; lines: the pre-registered regression). At equal perturbation size the deep loops cost $e^{\BnineDeepFp}$--$e^{\BtwoDeepFp}$
times more. (c)~Distance of loop $r$'s prefix keys to loop $R$'s: filled fp32, open bf16; solid $c\rho^r$ fitted to fp32,
dotted with one bf16 floor (Eq.~\eqref{eq:floor}). Raven's deep loops (shaded) sit at the bf16 floor; Huginn's do
not.}\label{fig:mech}
\end{figure*}

\subsection{Ripening: an approximate geometric trend}\label{sec:ripen}

Let $\Delta_r=\|K^{(r)}-K^{(R)}\|/\|K^{(R)}\|$ over prefix positions, averaged over layers and texts, and likewise for
values. Figure~\ref{fig:mech}c shows that in fp32, $\Delta_r\approx c\rho^r$ from loop 8 on (on Raven the rate drifts upward
with $r$; largest log residual \MResFpRavK). Huginn ripens at
$\rho=\MRhoHugK$ per loop (half-life \MHalfHugK{} loops) and Raven at $\rho=\MRhoRavK$ (\MHalfRavK{} loops); values
behave the same way ($\rho=\MRhoHugV$ and \MRhoRavV). In bf16 the curve flattens at a rounding floor,
\begin{equation}\label{eq:floor}
\Delta^{\mathrm{bf16}}_r\approx\sqrt{(c\rho^r)^2+\Delta_\infty^2},
\end{equation}
and this one extra parameter reproduces the bf16 curve ($\Delta_\infty=\MFloorHugK$ for Huginn's keys, \MFloorRavK{} for
Raven's). Over the deep loops 25--31, Huginn's keys are still \MDeepBfHugK{} away from plane $R$ in bf16, above its floor,
while Raven's are at its floor (\MDeepBfRavK); in fp32 Raven's are \MDeepRatioFpK{} times closer than Huginn's.

\textbf{Deepest-plane error.} Storing only plane $R$ sets $\delta_r=\Delta_r$ for every $r<R$. With the sensitivity profile of
Section~\ref{sec:sens}, we approximate the total by the deep-loop contribution, neglecting other loops and cross terms.
The fitted region model gives
\begin{equation}\label{eq:err}
\KL_{\mathrm{deep}}\approx e^{a+b}\,S^{g},\qquad S^2=\ncore\sum_{r=25}^{31}\big(\Delta^{K\,2}_r+\Delta^{V\,2}_r\big).
\end{equation}
For Huginn in fp32 this predicts \(\MErrPredHugFp\) at the first decoded position from the ripening curve, or
\(\MErrSizeHugFp\) from the measured perturbation size; the deepest-plane cache, which the fit did not use, measures
\(\MErrSHugFp\) (median). The agreement is to within a factor of two, not better. Substituting plane $R$ in the shallow loops
alone costs only \(\MErrShallowHugFp\). For Raven the model predicts \(\MErrPredRavFp\) (\(\MErrSizeRavFp\) from the measured
size), \MErrRatio--\MErrRatioSize{} times less than for Huginn, and the measurement is at the fp32 resolution (median below
$10^{-9}$, mean \(\MErrSmeanRavFp\)).

\begin{remark}[ripeness test, a hypothesis]
If the deep loops' $\Delta_r$ is at the bf16 floor, deepest-plane storage perturbs those loops about as much as rounding
already does, and its cost should be indistinguishable from a precision change; otherwise it should exceed it. Both models
fit this reading (Section~\ref{sec:fidelity}), but two models cannot establish it. The test needs one prefill per prompt
and no generation.
\end{remark}
\subsection{Acceptance}\label{sec:accept}

If each draft token is accepted with probability $\alpha$, independently of its position in the block, the tokens per
verification follow \citep{leviathan,chen}
\begin{equation}\label{eq:tau}
\tau(\alpha,\gamma)=\frac{1-\alpha^{\gamma+1}}{1-\alpha}.
\end{equation}
Batch 3 simulated greedy self-speculation on \BthreeN{} GSM8K problems with $\gamma=4$ and recorded acceptance at each
draft position given that all earlier ones were accepted. It is flat: across positions it varies by at most
\MSpreadFull{} (full cache) and \MSpreadDeep{} (deepest plane) for $d\in\{2,4,8\}$ (Figure~\ref{fig:spec}a), and
Eq.~\eqref{eq:tau} with the mean $\alpha_d$ matches the simulated $\tau$ within \MTauErrMax. Out of sample, the
batch-3 values predict $\tau=\MTauPredFullOos$ (full) and \MTauPredDeepOos{} (deepest plane) at $d=\gamma=4$; batch 8's
real batched-verify decoder measured \BeightTpvFull{} and \BeightTpvSfour{} on \BeightMetricsGsm{} other GSM8K problems and \MTauTimFull{}
and \MTauTimDeep{} on \BeightTPrompts{} more. Drafts that read the deepest plane accept $\alpha_4=\MAlphaDeepFour$ instead of
\MAlphaFullFour, and their rejection rate falls faster with draft depth: $1-\alpha_d\propto\rho_a^d$ with
$\rho_a=\MRejRhoDeep$ against \MRejRhoFull.

\subsection{Cost}\label{sec:cost}

We assume that at batch size 1 our eager decoder is bound by kernel launches and weight reads, so that a forward costs in
proportion to the layer passes it runs, $c(d)=\npre+d\,\ncore+\ncoda$, almost regardless of how many tokens it carries. A sequential token
costs $c(R)=\MCfull$ passes; a block costs $\gamma\,c(d)+c(R)$, scaled by one overhead factor $\eta$ for everything else.
The speed-up over sequential decoding is then
\begin{equation}\label{eq:speed}
S(d,\gamma)=\frac{\tau(\alpha_d,\gamma)\;c(R)}{\eta\,\big(\gamma\,c(d)+c(R)\big)}.
\end{equation}
With $\eta=1$ and the measured block counts, Eq.~\eqref{eq:speed} gives \MSpZeroFull$\times$ and \MSpZeroDeep$\times$
against measured \MSpMeasFull$\times$ and \MSpMeasDeep$\times$ (\MSpPrompts{} problems, $d=\gamma=4$). One factor
$\eta=\MEta$, fitted to all \MSpTimings{} per-problem timings at once, gives \MSpOneFull$\times$ and \MSpOneDeep$\times$, with
per-problem errors up to \MSpPromptErr. Sequential decoding is equally fast with either cache (ratio \BeightSeqRatio), so the gain of
the deepest plane is entirely its higher acceptance. With $\eta$ fixed at this value, Eq.~\eqref{eq:speed} predicts the two
\sysname{} decoders of Section~\ref{sec:ripe-e2e}, run later on another GH200 machine, at \BrPredRipe$\times$ and
\BrPredDd$\times$ over sequential \sysname; they measured \BrMeasRipe$\times$ and \BrMeasDd$\times$. Applied post hoc to
the full-cache and deepest-plane decoders of the same run, it is off by \BrAncErrFull{} and \BrAncErrSfour.

Figure~\ref{fig:spec}b extrapolates Eq.~\eqref{eq:speed} over $\gamma$ (flat acceptance was tested only to $\gamma=4$, and
only $d=\gamma=4$ was timed): deepest-plane drafts raise the speed-up and move the best block length from
$\gamma=\MDesignFullFourGamma$ (\MDesignFullFourS$\times$) past $\gamma=\MDesignDeepFourGamma$ (\MDesignDeepFourS$\times$).

\subsection{Memory and fixed-memory throughput}\label{sec:memory}

WaveFront's batched emulator \citep{wavefront} allocates exactly Eq.~\eqref{eq:mem}: over all \MCells{} cells of batches 6
and 8 (sequences of \MMaxSeq{} tokens; 0, 1 or 4 shared core slots), the reported cache size matches the formula to its
printed precision. Peak reserved memory adds $o_0+o_1B$, with $o_0=\MOzero$\,GB and $o_1=\MOoneFull$\,GB (full) or
\MOoneShared\,GB (one slot) per sequence, fitted on batch 6. This predicts the \MPeakOosN{} batch-8 cells, which ran on
another node and up to $B=256$, within \MPeakOosErr{}. Under the pre-registered 95\,GB budget, the largest batch is
$B_{\max}=\MBmaxFull$ with the full cache (\MPerSeqFull\,GB per sequence) and \MBmaxShared{} with one slot.

Throughput per configuration follows $T(B)=B/(a+bB)$ (largest residual \MTLossSixRes{} and \MTShSixRes); the full cache
saturates at $T_\infty=\MTLossSixInf$ tokens/s and one shared slot at \MTShSixInf{} (batch 6). The fixed-memory advantage of
sharing has two factors (Figure~\ref{fig:spec}c): at the full cache's best batch ($B=\MDecBeq$) the shared slot is already
\MDecEq$\times$ faster, and memory then lets it run $B=\MDecSharedB$, another \MDecBatch$\times$, for \MDecTotal$\times$ in all
(Section~\ref{sec:fixedmem}).

\section{The policy under test: \sysname}\label{sec:ripe}

\sysname{} (\emph{\sysexpansion}) reads the deepest plane for old positions and the own
plane for the $W$ most recent ones:
\begin{equation}
\pi_{\textsc{\sysname}}(r,t,j)=\begin{cases}r,& t-j\le W,\\ R,& t-j>W,\end{cases}
\end{equation}
\begin{equation}\label{eq:ripemem}
M_{\textsc{\sysname}}(L)=M_{\mathrm{deep}}(L)+b\,\ncore\,(R-1)\min(W,L).
\end{equation}
The accounting assumes a deepest-plane array for all positions plus a ring buffer of $W\times(R-1)$ core planes, which the reported GPU experiments do not allocate; a
position leaving the window drops its non-deepest planes. At $L=1024$, $W=64$ costs \MHugRipeSixtyFourMB\,MB on Huginn
(\MHugRipeSixtyFourFrac{} of \MHugFullMB\,MB) and \MRavRipeSixtyFourMB\,MB on Raven (\MRavRipeSixtyFourFrac), i.e.\ about
\MHugRipeSixtyFourExtra\,MB per Huginn sequence over the deepest-plane cache (Table~\ref{tab:mem}). These are calculated
values: the Huginn decoder emulates storage (Section~\ref{sec:ripe-e2e}), while the Ouro runs use a physical ring buffer
(Section~\ref{sec:ouro}; CPU checks in the supplementary appendix).

\begin{table}[t]\centering\small
\caption{KV memory per sequence at $L=1024$ from Eqs.~\eqref{eq:mem} and~\eqref{eq:ripemem} (bf16, MB $=10^6$ bytes), and
its fraction of the full cache.}\label{tab:mem}
\begin{tabular}{@{}lrrrr@{}}\toprule
 & \multicolumn{2}{c}{Huginn-0125} & \multicolumn{2}{c}{Raven-Llama}\\\cmidrule(lr){2-3}\cmidrule(l){4-5}
Cache & MB & frac. & MB & frac.\\\midrule
full & \MHugFullMB & 100\% & \MRavFullMB & 100\%\\
deepest plane & \MHugDeepMB & \MHugDeepFrac & \MRavDeepMB & \MRavDeepFrac\\
two planes & \MHugTwoMB & \MHugTwoFrac & \MRavTwoMB & \MRavTwoFrac\\
\sysname, $W=16$ & \MHugRipeSixteenMB & \MHugRipeSixteenFrac & \MRavRipeSixteenMB & \MRavRipeSixteenFrac\\
\sysname, $W=64$ & \MHugRipeSixtyFourMB & \MHugRipeSixtyFourFrac & \MRavRipeSixtyFourMB & \MRavRipeSixtyFourFrac\\
\sysname, $W=256$ & \MHugRipeTwoFiveSixMB & \MHugRipeTwoFiveSixFrac & \MRavRipeTwoFiveSixMB & \MRavRipeTwoFiveSixFrac\\\bottomrule
\end{tabular}
\end{table}

\begin{algorithm}[t]
\caption{One \sysname{} self-speculative step (greedy)}\label{alg:ripe}
\begin{algorithmic}[1]
\STATE {\bfseries Input:} last token $x$ at position $t_0$ (not yet cached); $\mathcal C$ for positions $j<t_0$ (plane $R$
everywhere, all planes for $t_0-j\le W$); draft depth $d$, block length $\gamma\le W$
\FOR{$i=1$ {\bfseries to} $\gamma$}
\STATE $\hat y_i\leftarrow\arg\max$ of a $d$-loop pass on $\hat y_{i-1}$ ($\hat y_0=x$); read plane $R$ before $t_0$,
own-loop planes within the draft chain
\ENDFOR
\STATE $z\leftarrow$ one $R$-loop pass on $[x,\hat y_1,\dots,\hat y_\gamma]$; loop $r$ of token $t$ reads plane $r$ of $j$ if
$t-j\le W$, else plane $R$
\STATE $k\leftarrow\max\{i:\hat y_{1..i}=\arg\max z_{0..i-1}\}$; emit $\hat y_{1..k}$ and $\arg\max z_k$
\STATE commit verified KV for $[x,\hat y_1,\dots,\hat y_k]$ only; invalidate the remaining draft suffix
\STATE $x\leftarrow\arg\max z_k$, $t_0\leftarrow t_0+k+1$; keep $x$ uncached; drop non-$R$ planes where $t_0-j>W$
\end{algorithmic}
\end{algorithm}
Algorithm~\ref{alg:ripe} describes a nonterminal step; stopping rules may truncate the emitted sequence.
The correction or bonus token is processed at the next step. In the emulator, invalidation is logical: the next write
position excludes the unused suffix from attention, and later passes overwrite it.

\begin{proposition}[sequential self-consistency under \sysname]\label{prop:bv}
Consider a causal model whose operations use the same per-position initialization and do not couple a token to future
positions. Let the committed prefix state be identical in both executions, stochastic layers be disabled, and the block
pass retain every prefix plane needed by its queries until those queries finish. If $W\ge\gamma$, verification of a fixed
block of $\gamma+1$ tokens under $\pi_{\textsc{\sysname}}$ computes, in exact arithmetic, the same logits and cached keys and
values as processing that block sequentially under the same policy.
\end{proposition}
\begin{proof}[Proof sketch]
Use induction over recurrent loops and, within each loop, causal positions. The two schedules start from the same prefix
and per-position state. For a block query at $t$, each earlier block position $j$ has $t-j\le\gamma\le W$, so both
schedules read its loop-$r$ plane; it is available from the current loop in the block pass and from the cache in the
sequential pass. A committed prefix position $j<t_0$ read at loop $r$ with $t-j\le W$ also satisfies $t_0-j\le W$, so its
loop-$r$ plane was retained at block entry. Otherwise both schedules read its stored plane $R$. By the induction
hypothesis the inputs to each causal operation agree, so its output and emitted K/V agree. The coda then produces equal
logits. In particular, the committed K/V for an accepted prefix agree.
\end{proof}
This compares execution schedules of the same approximate cache; it says nothing about their distance from full-cache
decoding. In floating point, kernel shapes and reduction orders change logits, so it does not guarantee bitwise identity.

The deepest-plane cache ($W=0$) violates the premise: a block token cannot read plane $R$ of an earlier block token, which
does not exist yet, so batched verification reads plane $r$ there, and outputs depend on where blocks start.
Section~\ref{sec:spec} measures how much this matters on Huginn.

\textbf{Drafts and the choice of $W$.} Under the proposition's assumptions, greedy verification preserves the sequential
trajectory of the same \sysname{} target regardless of the draft tokens; this is not a full-cache equivalence claim. Reading
plane $R$ gives drafts the higher acceptance of Section~\ref{sec:accept}. The block length sets $W\ge\gamma$; we use $W=64$, the
window that reduced Huginn's replay tail (Section~\ref{sec:fidelity}), as the tested operating point.

\section{Evaluation}\label{sec:eval}

\textbf{Protocol.} In the experimental records, each batch had a pre-registration with decision rules and predictions, a
frozen file table (SHA-256), on-node CPU self-tests, and an archive verified on both ends. Separate unscored Nsight Compute
passes were planned; batch R's pass was skipped. All runs used one NVIDIA GH200 (96\,GB), batch R on a second GH200 machine, with PyTorch 2.14 and CUDA 13.0,
bf16 and greedy decoding unless stated, and cuDNN attention disabled from batch 7 on (Section~\ref{sec:repro}). Prompts
are GSM8K 8-shot problems \citep{gsm8k} and AlpacaEval instructions or reference texts \citep{alpacaeval}. Huginn uses HF
revision bb6621b6 and Raven revision a5f6f126, both with the official modeling code unchanged. The decoder timings below
use one request at a time; queueing, continuous batching and cache allocation are not measured.
The supplementary appendix lists the batches and the artifact's scope. Frozen inputs and a reconstructed runtime
recipe now accompany the Huginn decoder and window-replay entry points. The pre-registration files remain outside
the snapshot, so claims about their timing rely on the authors' record.
Figure~\ref{fig:experiment-overview} summarizes the three measurements that frame the rest of this section.

\begin{figure*}[t]
\centering
\includegraphics[width=\textwidth]{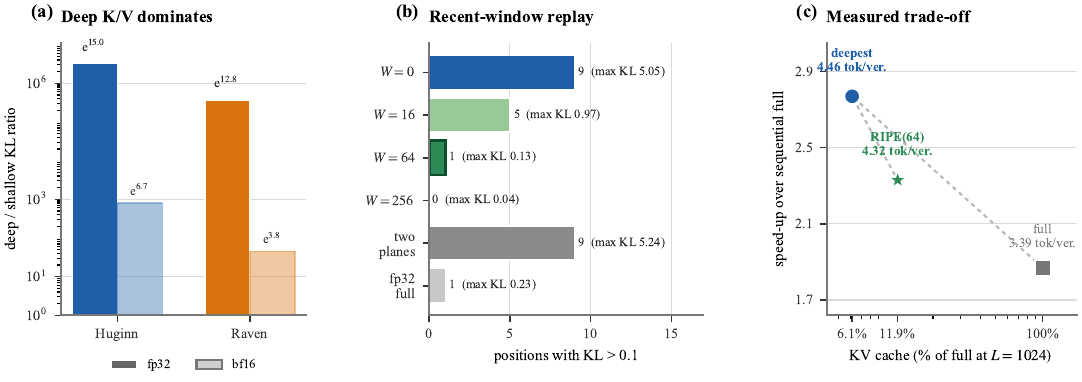}
\caption{Experimental summary. (a) At matched perturbation size, deep-loop K/V is far more sensitive than shallow K/V;
the bf16 floor compresses the gap. (b) In 48-prompt Huginn replay, a 64-position window reduces the number of positions
with KL $>0.1$ from 9 to 1, whereas storing two planes does not. (c) Measured batch-size-one speed-up against calculated
cache footprint. \sysname(64) retains most of the deepest-plane acceptance, but is slower and uses more cache; on the
22-prompt decoder run, its target is not measurably closer to full-cache generation than the deepest plane.}
\label{fig:experiment-overview}
\end{figure*}

\subsection{Generation fidelity and the tail}\label{sec:fidelity}

\begin{figure*}[t]
\centering
\includegraphics[width=\textwidth]{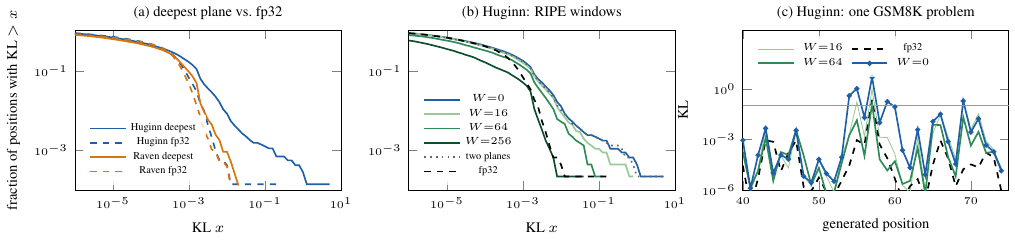}
\caption{Error of deepest-plane caches in generation (teacher-forced replay of full-cache trajectories). (a)~Survival
function of per-position KL, \MTailPositions{} Huginn and \MRavPositions{} Raven positions on the same 72 prompts.
(b)~\sysname{} windows on the 48 Huginn prompts of batch 8. (c)~KL along the GSM8K problem that holds \MTailTop{} of
Huginn's \MTailN{} tail positions.}\label{fig:tail}
\end{figure*}

\textbf{Huginn: mostly a precision change, plus a tail.} The pre-registered statistic is
$\KL(L_0\|\text{deepest})/\KL(L_0\|\text{fp32})$ per bin of generated position with a prompt-paired bootstrap. On 48
confirmatory prompts the label is ``\BsevenConeLabel'', through a single bin (ratio \BsevenConePeak{}
[\BsevenConePeakLo, \BsevenConePeakHi]) in which one problem carries \BsevenPeakShare{} of the KL; other bins have ratios
\BsevenBinMin--\BsevenBinMax. Over all positions the median KL is lower for the deepest plane than for fp32
(\(\BsevenMedSfour\) vs.\ \(\BsevenMedFp\)) but the mean is \BsevenMeanRatio{} times higher, because 1\% of positions carry
\BsevenShareSfour{} of the KL (fp32: \BsevenShareFp; Figure~\ref{fig:tail}a). The \MTailN{} positions with KL above 0.1
(fp32: \MTailFpHug) lie in \MTailPrompts{} prompts, \MTailTop{} of them in one GSM8K problem at nearby positions
(Figure~\ref{fig:tail}c). Only \MTailLowMargin{} of them have a reference top-2 probability margin below 0.1, where
\MLowMarginFrac{} of all positions lie: near-ties are over-represented (\MTailLowExp{} expected, \MTailLowOver$\times$)
but account for a minority of the tail. The largest value, \BsevenMaxSfour, is one of those two.

\textbf{Raven: indistinguishable from a precision change.} On the same 72 prompts, Raven's deepest plane is
``\BnineConeLabel'' on both prompt sets (bin ratios \BnineBinMin--\BnineBinMax), its largest KL is \BnineMaxSsix{}
(fp32: \BnineMaxFp), no position exceeds 0.1, and it flips the top-1 token at \BnineFlipSsix{} positions against
\BnineFlipFp{} for fp32. We had predicted the opposite, a larger effect than on Huginn; the ripening model accounts for the
difference after the fact (Section~\ref{sec:ripen}).

\textbf{\sysname{} windows remove most of Huginn's tail.} Batch 8 replayed 48 of these prompts through \sysname{} caches
(Table~\ref{tab:window}, Figure~\ref{fig:tail}b). The pre-registered count of tail positions falls from \BeightTailC{} to
\BeightTailWsixtyfour{} at $W=64$ (\BeightBoneRed{} fewer; ratio \BeightBoneRatio{} [\BeightBoneLo, \BeightBoneHi], wide
because four prompts carry all tail positions), and the largest KL from \BeightMaxC{} to \BeightMaxWsixtyfour, below the
fp32 replay's \BeightMaxF. On the problem of Figure~\ref{fig:tail}c, $W=16$ already cuts the peak to \MTraceWsixteen{} and
$W=64$ to \MTraceWsixtyfour. Keeping four attention-sink positions \citep{attentionsinks} at full depth does nothing, and
neither does a second stored plane at similar memory. A post hoc prompt-level audit finds that the affected prompts fall
from \ReplayTailPromptsC/\BeightBPrompts{} to \ReplayTailPromptsWsixtyfour/\BeightBPrompts{};
leaving out any one prompt gives a tail-count reduction of \ReplayLooMin--\ReplayLooMax.
This is a sensitivity diagnostic, not a confidence interval or evidence of generalization.
The \BeightBPrompts{} reference trajectories were bitwise identical
to batch 7's, on another node. This replay result is what motivated the window; Section~\ref{sec:ripe-e2e} tests whether it
survives in a decoder.

\begin{table}[t]\centering\small
\caption{\sysname{} windows on Huginn (batch 8, \BeightBPrompts{} prompts, \ReplayPositions{} positions).
Tail: positions with KL$>0.1$; texts: prompts with any tail event. Extra: memory over
the deepest-plane cache at $L=1024$ (Eq.~\eqref{eq:ripemem}).}\label{tab:window}
\setlength{\tabcolsep}{5pt}
\begin{tabular}{@{}lrrrrr@{}}\toprule
Cache & Tail & Texts & Max KL & Flips & Extra MB\\\midrule
deepest ($W=0$) & \BeightTailC & \ReplayTailPromptsC & \BeightMaxC & \BeightFlipC & 0\\
\sysname, $W=16$ & \BeightTailWsixteen & \ReplayTailPromptsWsixteen & \BeightMaxWsixteen & \BeightFlipWsixteen & \MHugRipeSixteenExtra\\
\sysname, $W=64$ & \BeightTailWsixtyfour & \ReplayTailPromptsWsixtyfour & \BeightMaxWsixtyfour & \BeightFlipWsixtyfour & \MHugRipeSixtyFourExtra\\
\sysname, $W=256$ & \BeightTailWtwofivesix & \ReplayTailPromptsWtwofivesix & \BeightMaxWtwofivesix & \BeightFlipWtwofivesix & \MHugRipeTwoFiveSixExtra\\
$+$ 4 sinks & \BeightTailKfour & \ReplayTailPromptsKfour & \BeightMaxKfour & \BeightFlipKfour & \MHugSinkExtra\\
two planes (\texttt{s8}) & \BeightTailSeight & \ReplayTailPromptsSeight & \BeightMaxSeight & \BeightFlipSeight & \MHugTwoExtra\\
fp32 full & \BeightTailF & \ReplayTailPromptsF & \BeightMaxF & \BeightFlipF & ---\\\bottomrule
\end{tabular}
\end{table}

\subsection{Self-speculative decoding}\label{sec:spec}

\begin{figure*}[t]
\centering
\includegraphics{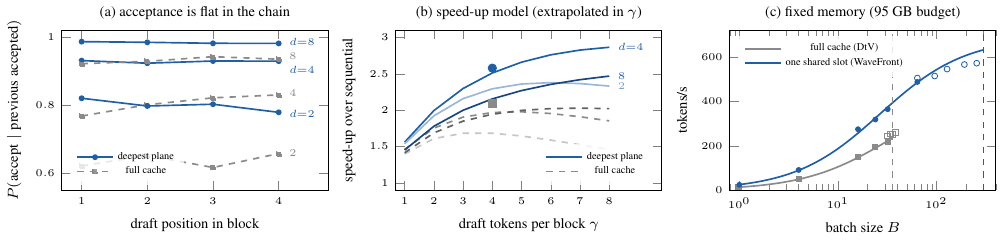}
\caption{(a)~Acceptance at each draft position given that earlier ones were accepted (batch 3, Huginn, GSM8K); flatness
supports Eq.~\eqref{eq:tau}. (b)~Eq.~\eqref{eq:speed} with $\eta=\MEta$ (lighter: $d=2$, darker: $d=8$; dashed: full cache); markers: measured speed-ups of the real decoder at
$d=\gamma=4$ (batch 8). Curves beyond $\gamma=4$ are extrapolations. (c)~Throughput of WaveFront's batched emulator at
measured acceptance rates; filled: batch 6, open: batch 8 (another node, reported separately); lines: $T(B)=B/(a+bB)$ fitted
to batch 6, drawn up to the memory model's $B_{\max}$ at 95\,GB (dashed).}\label{fig:spec}
\end{figure*}

Batch 8 ran a real greedy self-speculative decoder on Huginn: depth-4 drafts, four per block, and one 32-loop verification
pass per block. With the full cache it accepts \BeightTpvFull{} tokens per verification on GSM8K and is
\BeightSpeedFull$\times$ faster than sequential decoding; with the deepest plane, \BeightTpvSfour{} and
\BeightSpeedSfour$\times$ (Sections~\ref{sec:accept}--\ref{sec:cost} model both).

With the deepest plane, batched verification depends on where blocks start (Section~\ref{sec:ripe}). The pre-registered
test compares verification targets of batched and sequential deepest-plane decoding at identical contexts, against the
same comparison for the full cache, where the two are equal in exact arithmetic. A tiny CPU model checks the implementation
within numerical tolerance; the GPU comparison supplies the bf16 numerical baseline. The ratio is \BeightAoneRatio{} [\BeightAoneLo, \BeightAoneHi], below the pre-registered threshold
of 2 (label ``\BeightAoneLabel''; we had predicted 3--30), and \BeightBvSfourSame{} of \BeightAPrompts{} batched
trajectories equal the sequential ones. On Huginn the block dependence is therefore small; Proposition~\ref{prop:bv}
removes it by construction.

Using WaveFront's released code \citep{wavefront} unchanged, batch 5 reproduced its Huginn numbers on GH200: speed-up
\BfiveSpeedZero$\times$ at acceptance \BfiveAccZero{} without sharing and \BfiveSpeedOne$\times$ at \BfiveAccOne{} with one
shared slot.

\subsection{\sysname{} under emulated storage}\label{sec:ripe-e2e}

\begin{table}[t]\centering\small
\caption{\sysname(64) inside a batch-size-one batched-verify decoder (batch R, Huginn, GSM8K, $d=\gamma=4$, another
GH200 machine). Tokens per verification from \BrPrompts{} problems; speed-up over sequential full-cache decoding from
\BrTimingPrompts{} problems run alone on the GPU. \sysname's speed-ups include the cost of emulating its storage but not a
measured reduction in cache allocation (Section~\ref{sec:ripe-e2e}).}\label{tab:ripe-e2e}
\begin{tabular}{@{}lcc@{}}\toprule
Decoder & Tokens/verify & Speed-up\\\midrule
full cache & \BrTauFull & \BrSpeedFull$\times$\\
deepest plane ($W=0$) & \BrTauSfour & \BrSpeedSfour$\times$\\
\sysname(64), drafts follow $\pi_{\textsc{\sysname}}$ & \BrTauRipe & \BrSpeedRipe$\times$\\
\sysname(64), drafts read plane $R$ & \BrTauDd & \BrSpeedDd$\times$\\\bottomrule
\end{tabular}
\end{table}

Batch R implements Algorithm~\ref{alg:ripe} with a per-query window and runs it, pre-registered, on the prompts of batch~8
(Table~\ref{tab:ripe-e2e}). Storage is emulated: the cache still holds every plane and attention reads them by
$\pi_{\textsc{\sysname}}$, so the batch measures decoding behaviour and batch-size-one timing, not memory or serving
throughput (Eq.~\eqref{eq:ripemem} gives the calculated memory).
On a tiny CPU model, batched and sequential \sysname{} produce identical token trajectories when $W\ge\gamma$;
their stored log-probabilities agree within the test's $10^{-10}$ absolute tolerance. The tests use fp64 model arithmetic
but cast logits to fp32, so they check the implementation rather than prove exact equality. A $W<\gamma$ case gives
different rows. On the GPU the observed difference is compatible with the bf16 numerical baseline: the batched-versus-sequential
ratio is \BrOneRatio{} [\BrOneLo, \BrOneHi] times the full cache's floor (label ``\BrOneLabel''), and \BrDdEqSeq{} of
\BrPrompts{} trajectories equal sequential \sysname.

Drafts decide the acceptance. Drafts that follow $\pi_{\textsc{\sysname}}$ read shallow, unripe planes for the recent window and
accept like the full cache (\BrTauRipe{} tokens per verification); drafts that read plane $R$ keep \BrRecDd{}
[\BrRecDdLo, \BrRecDdHi] of the deepest plane's gain (\BrTauDd{} against \BrTauSfour, label ``\BrRecDdLabel''). That decoder
runs \BrSpeedDd$\times$ faster than sequential full-cache decoding. Sequential \sysname{} runs at \BrSeqRatio{} of
sequential full-cache speed because the emulation assembles planes on every call; against it, the decoder measures
\BrMeasDd$\times$, where the cost model of Section~\ref{sec:cost} predicts \BrPredDd$\times$ (pre-registered check). Post
hoc, the same model predicts the full-cache and deepest-plane decoders of this run at \BrAncPredFull$\times$ and
\BrAncPredSfour$\times$ over sequential full-cache decoding against measured \BrAncMeasFull$\times$ and
\BrAncMeasSfour$\times$, so its error on this machine is up to \BrAncErrMax. Along their own trajectories, the verification targets of \sysname(64) and of the deepest plane are
not distinguishably different in their distance from the full cache on these \BrPrompts{} problems (ratio \BrThreeRatio{} [\BrThreeLo, \BrThreeHi], label
``\BrThreeLabel''; we had predicted closer), with \BrTailDd{} against \BrTailSfour{} positions above KL 0.1 but a larger
maximum (\BrMaxDd{} against \BrMaxSfour). \sysname{} is also slower (\BrSpeedDd$\times$ versus \BrSpeedSfour$\times$) and
larger in calculated storage (\MHugRipeSixtyFourMB\,MB versus \MHugDeepMB\,MB at $L=1024$).
The confidence interval includes both a reduction and an increase; it does not establish equivalence. These KL values
compare each arm to a full-cache reference on that arm's own contexts, not on one common trajectory, and are not a
task-accuracy equivalence test. The replay tail reduction therefore has not translated into an established end-to-end
quality gain: on Huginn the window is ineffective as a repair, and it costs speed (\BrSpeedDd$\times$ against
\BrSpeedSfour$\times$) and \MHugRipeSixtyFourExtra\,MB per sequence of calculated storage. The run stopped after \BrPrompts{} of \BrPlanned{} problems,
before AlpacaEval, and skipped its Nsight Compute pass.

\subsection{Pretrained Ouro: the prompt lies outside the window}\label{sec:ouro}

\begin{table}[t]\centering\small
\caption{Ouro-1.4B, 24 reused GSM8K questions with the chat template, greedy, stop only at EOS or 2048 tokens; each
question's four methods ran on one GPU (12 questions on two L40S, 12 on two RTX PRO 6000). Repeating: capped outputs whose
last 512 tokens have eight-gram diversity below 0.05. Last number: post hoc, EOS outputs only; not an accuracy over the
sample. Memory budgets are not matched.}\label{tab:ouro}
\setlength{\tabcolsep}{3.5pt}
\begin{tabular}{@{}lcccc@{}}\toprule
Cache & EOS/24 & Repeat. & Med.\ tokens & Last no.\\\midrule
full & 23 & 0 & 533 & 21/23\\
\sysname, $W=16$ & 0 & 23 & 2048 & ---\\
\sysname, $W=64$ & 2 & 19 & 2048 & 2/2\\
FlashLoop (KIVI) & 23 & 0 & 515 & 23/23\\\bottomrule
\end{tabular}
\end{table}

Ouro-1.4B \citep{ouro} has $R=4$ loops over 24 layers and its own report on KV sharing: for decoded tokens, keeping the last
loop's KV is nearly lossless, but sharing it across the prompt costs more than ten GSM8K points. We ran \sysname{} on the
pinned official Ouro classes through a physical ring-buffer implementation of Eq.~\eqref{eq:ripemem} (supplementary appendix),
first on a fixed 20-case numerical qualification and then on a task diagnostic. In the qualification, physical and emulated
storage agree bitwise in all 16 \sysname{} cases on an L40S and again on an RTX PRO 6000, so the implementation realizes the policy.
The sequential-versus-batched verification gate, however, fails in all 16 bf16 cases on both GPUs (argmax disagreements in
two) and passes in all 16 when the same plan is rerun in fp32 with TF32 disabled, which places the discrepancy with bf16
execution shape rather than with the cache (Section~\ref{sec:repro}).

The task diagnostic (Table~\ref{tab:ouro}) re-encodes 24 GSM8K questions that earlier batches had already used and lets each
method generate until EOS or 2048 tokens. The full cache and FlashLoop's default KIVI reader \citep{flashloop} finish 23 of 24
answers; \sysname{} finishes none at $W=16$ and two at $W=64$, one of them a 38-token answer identical to the full cache's,
and 23 and 19 of its capped outputs end in near-verbatim repetition. Replaying the first question's capped $W=16$ and $W=64$
outputs teacher-forced through physical and emulated storage gives bitwise equal logits at all 2048 positions, so the
repetition is the policy's, not a ring-buffer defect (an error in the substitution code common to both paths is not excluded).
The explanation is positional: the chat-template prompts have 318--415 tokens, so with $W\le64$ at least \MOuroPromptShared{}
of each prompt is read at the deepest plane, which is exactly the prompt-side sharing that Ouro's authors found harmful. The window
protects the positions being decoded, but on this model those are not the positions that matter. Earlier 160- and 512-token
runs of the same questions are retained as diagnostics; none of these runs is a matched-memory or accuracy benchmark, and a
variant that keeps every plane of the prompt was not tested.

\subsection{Fixed-memory throughput}\label{sec:fixedmem}

Batches 6 and 8 placed WaveFront's scheduler and draft-then-verify under a 95\,GB budget at the measured acceptance rates
(Figure~\ref{fig:spec}c): one shared slot is \BsixRatio$\times$ (batch 6) and \BeightDratio$\times$ (batch 8, another node) faster
than the best full-cache configuration, through a larger batch and a cheaper step (Section~\ref{sec:memory}). Two pre-registered memory
predictions failed; the memory model reproduces both cells. On Huginn a one-slot cache is not lossless (Section~\ref{sec:fidelity}),
and \sysname{} was not run under this scheduler, so these figures are context for the memory stakes, not a measurement of \sysname{}.

\subsection{What did not help}\label{sec:negative}

\textbf{Per-head position filtering.} The starting plan also proposed keeping, per head, only the prefix positions that a
shallow loop attends to. Batch 1 found cross-depth attention affinity real but mostly static (sink plus recent tokens).
Batch 8 froze a per-head policy from batch 1 and tested it on 64 unseen AlpacaEval texts at a budget of 10\% of the prefix:
filtering costs \BeightConeRatio{} times more KL than deepest-plane storage, the two costs add (ratio \BeightCthreeRatio),
and the per-head policy is \BeightCtwoRatio$\times$ worse than top-$k$ for every head [\BeightCtwoLo, \BeightCtwoHi].
Selecting at loop 4 instead of 2 cuts the cost to \BeightCfourRatio{} of its value, consistent with ripening.

\textbf{Sinks and a second plane.} Neither removed any of Huginn's tail positions (Table~\ref{tab:window}); a recent window did.

\subsection{Reproducibility on GH200}\label{sec:repro}

Batch 4 found that the bf16 full-cache reference was not reproducible: a same-process rerun matched bitwise at only
\BfourRerunSame{} of positions, with mean KL \(\BfourRerunKL\), about two thirds of the bf16-to-fp32 floor. PyTorch's
deterministic mode did not remove it; batch 6b traced it to cuDNN's scaled-dot-product attention (cuDNN 9.24). With that
backend disabled, reruns are bitwise identical (\BsevenSamePos{} of positions, across processes and nodes) and decoding is
not slower; the deepest-plane labels and ratios were unchanged. Separately, WaveFront's decoder and its autoregressive
baseline produce identical texts on \BsevenWsameZero/40 problems without sharing, with cuDNN disabled or not, and our own
full-cache batched and sequential decoders, equal in exact arithmetic, agree on only \BeightBvfSame{} of \BeightAPrompts{}
trajectories in bf16. Computation shape alone flips bf16 trajectories at this rate, so a ``lossless'' claim for speculative
decoding in bf16 needs a stated reference computation.

This result is limited to the tested GH200 environments; the released \texttt{nocudnn} setting does not pin one attention
kernel, and our records do not establish cross-device bitwise equality. A replication should match model, tokenized
inputs, per-position initialization and software versions, then check same-device repeatability before comparing cache
fidelity or timing against a sequential baseline on the same device.

\section{Related work}\label{sec:related}

\textbf{Looped models.} Huginn \citep{huginn}, Ouro \citep{ouro} and retrofitted recurrence \citep{retrofit} train or convert
models that iterate a core block. Ouro reports that keeping only the last loop's KV for decoded tokens is nearly lossless on its model, whereas
sharing it across the prompt costs more than ten GSM8K points; we measure when deepest-plane reuse holds and why it fails on Huginn. WISE \citep{wise} finds that attention routing stabilizes early in depth
and reuses the early loops' attention support; we find that later loops are more sensitive to prefix keys and values and reuse the last loop's
planes. Both exploit convergence across depth, in opposite directions.

\textbf{Self-speculation in looped models.} Huginn's zero-shot self-speculation \citep{huginn}, LoopSpec's pipelined branches
\citep{loopspec}, WaveFront's diagonal schedule with cross-recurrence KV sharing \citep{wavefront}, Mature-Spec's mature
values for shallow drafts \citep{maturespec} and LoopAhead \citep{loopahead}. Self-speculation by layer skipping in standard
models is Draft\&Verify \citep{draftverify}; the acceptance model is that of \citet{leviathan} and \citet{chen}.

\textbf{KV compression.} For standard transformers: attention sinks and streaming windows \citep{attentionsinks},
quantization \citep{kivi} and paging \citep{pagedattention}. Along the loop axis, LLA learns a latent codec \citep{lla},
ResidualQuant quantizes loop residuals \citep{residualquant}, FlashLoop combines sparse updates with residual quantization
\citep{flashloop}, and HLA keeps a recent exact window and reads older positions through trained latents \citep{hla}. The
recent-window pattern is thus shared prior art; \sysname{} is its training-free member, and our result is that the window alone
does not repair the deepest plane. These methods' compression levels come from different models, metrics and stacks, and we
do not compare them at matched memory; the one same-model comparison we have (Table~\ref{tab:ouro}) is unmatched in memory and
shows FlashLoop finishing answers that \sysname{} does not.

\section{Limitations and conclusion}\label{sec:limits}

\textbf{Limitations.} The negative result rests on three models and small samples: \BrPrompts{} GSM8K problems in the
Huginn decoder run, \BeightBPrompts{} prompts in its replay with four prompts carrying the tail, 72 prompts on Raven, and 24
reused questions on Ouro. Only greedy decoding on GSM8K and AlpacaEval is covered, and the Ouro task runs are a diagnostic with
unmatched memory budgets, not an accuracy benchmark; an unfinished output is not established wrong. Batch R emulates storage, so
its memory is calculated and its timing is batch-size one; no serving throughput for \sysname{} is measured. We did not test a
window that also keeps the prompt's planes, nor $W$ between 64 and the prompt length, so the result bounds recent-window repair
at the sizes a ring buffer makes cheap, not every windowed policy. The ripening and sensitivity models are post hoc and fitted
on Huginn and Raven only; for Ouro the explanation rests on prompt length and Ouro's own report. The Ouro and rerun evidence
comes from rented L40S and RTX PRO 6000 GPUs with regenerated trajectories, so cross-device bitwise equality is not
established. Matched-memory comparisons with LLA, HLA, ResidualQuant, FlashLoop and WaveFront sharing remain open, and raw
outputs of the earlier batches are not released (supplementary appendix).

\textbf{Conclusion.} A recent full-depth window is the natural repair for a deepest-plane KV cache, and on the model
that motivated it, it repairs the metric that motivated it, not decoding. On Raven the deepest plane needs no repair; on
Huginn the repaired cache is slower, larger and no closer to the full cache in a real decoder; on Ouro the positions that need
repair are the prompt's, beyond any 16--64-position window. Ripeness, measurable with one prefill per prompt, and the role of
the shared positions decide the outcome; recency does not. Schedule consistency, which the window does provide, is not
fidelity, and in bf16 neither is bitwise agreement. Shared-plane caches for looped models should be qualified by a ripeness
test and a prompt-side check before any windowed repair is considered.

\clearpage
\onecolumn
\appendix
\section{Batches}\label{app:batches}

In the experimental records, each batch had a pre-registration with decision rules and predictions frozen before the
run, a frozen file table (SHA-256), on-node CPU self-tests, and an archive verified at both ends (Table~\ref{tab:batches}).
Separate unscored Nsight Compute passes were planned; batch R's pass was skipped.
Appendix~\ref{app:artifact} distinguishes the released measurement subset
from the unavailable records; the pre-registration documents remain outside the snapshot.

\begin{table}[H]\centering\small
\caption{Pre-registered batches, one GH200 each.}\label{tab:batches}
\begin{tabular}{@{}lp{7.5cm}p{6.5cm}@{}}\toprule
Batch & Question & Pre-registered label\\\midrule
1 & cross-depth attention affinity & GSM8K trivial; AlpacaEval intermediate\\
2 & which loops' prefix KV matter (Huginn) & deep KV matters\\
3 & deepest plane in generation; acceptance & exceeds precision change\\
4 & reference asymmetry; confirmation; tail & asymmetry negligible; exceeds\\
5 & WaveFront baseline & speed-ups in band\\
6, 6b & reproducibility; fixed-memory frontier & cuDNN is the source; shared $\ge1.5\times$\\
7 & batches 3--5 with cuDNN off & reproducible; unchanged\\
8 & batched verification; windows; per-head filtering; frontier & as numerics; window removes most; per-head worse\\
9 & Raven-Llama: region test; deepest plane vs.\ fp32 & deep KV matters; indistinguishable\\
R & \sysname(64) in the batched-verify decoder (second GH200 machine) & as numerics; keeps the gain; fidelity no difference\\\bottomrule
\end{tabular}
\end{table}

\section{Artifact scope and reproduction status}\label{app:artifact}

The repository snapshot contains the 42 Python files from the frozen research batch, the pinned Huginn and Raven
tokenizers, the model source retained from the upstream revisions, CPU tests and 12 original input JSON files.
All 42 Python files, six tokenizer files and 12 inputs match the historical batch-R SHA-256 table. The inputs preserve
token IDs, seeds, generation limits, stopping tokens and ordered cohort selections, so executing these inputs does not
require rebuilding prompts or re-tokenizing a dataset. The snapshot also retains the earlier Huginn and Raven input
sets; the packaged GPU workflow targets the two Huginn experiments below.

\textbf{Environment and execution.} The recorded batch-R runtime was Python 3.12, PyTorch 2.14.0 with CUDA 13.0 and
Transformers 4.57.6. The artifact includes that package inventory and a reconstructed, hash-locked 64-package
installation for Linux x86\_64/Python 3.12. This is an installation recipe, not a historical container or disk image;
the original Python patch version was not recorded. A downloader fetches only the six required Huginn model files
from revision bb6621b6 and checks their historical digests. Weights are not stored in the repository. The supplemental
\texttt{scripts/reproduce\_gpu.py} provides CPU integrity checks, a nonexecuting command plan, a one-prompt smoke run,
the complete batch-R cohort and the complete batch-8 window replay. Setup and model downloads precede GPU execution.

\begin{table}[H]\centering\small
\caption{Packaged Huginn reproduction workloads. These are executable protocols, not new measured results.}\label{tab:repro-workloads}
\begin{tabularx}{\textwidth}{@{}lp{4.0cm}X@{}}\toprule
Entry point & Fixed workload & Comparison supported\\\midrule
\texttt{smoke} & First original metric prompt; unchanged generation limit & Startup and execution of all metric arms; not a full replication\\
\texttt{batch-r} & 32 metrics (24 GSM8K, 8 chat); 12 GSM8K timings & Same-device fidelity, acceptance and timing; original $W=64$, $d=\gamma=4$\\
\texttt{window-replay} & 48 prompts (36 GSM8K, 12 chat); all eight comparison arms & Tail behavior on a common full-cache trajectory generated on the new device\\\bottomrule
\end{tabularx}
\end{table}

Each invocation requires a fresh output directory and records commands, package versions, file hashes, device
properties and original JSONL outputs. Completion checks require the exact ordered prompt lists, metadata,
configuration, arms and finite per-position metrics. Timeout or failure preserves partial logs and a failure receipt;
the archive is independently read back and hashed. A completed case list does not imply that a scientific prediction
held. Environment deviations require an explicit override and are recorded as a separate cohort. CPU tests, target
dependency resolution and an installation dry run have been checked; a clean Linux installation and H200 full-model
execution have not been independently validated by this artifact release.

\textbf{Comparable cohorts.} The original decoder run completed 22 GSM8K prompts, whereas its intended 32-prompt
cohort also contains chat prompts. The new-run analysis therefore reports the full cohort, each domain and the fixed
22 IDs corresponding to the published records separately. A matched-prompt cross-device comparison uses those 22
IDs; full-cohort results must also be retained. These outputs are neither pooled with the historical cohort nor used
to retroactively complete its interrupted run. The 12 timing prompts are measured against a sequential baseline on
the same device. In window replay, a different full-cache trajectory on another device can change the positions being
compared; common-context diagnostics are needed before attributing a discrepancy to cache sharing.

\textbf{L40S batch-R rerun.} A separate shared-host L40S execution completed all 12 timing and 32 metric prompts.
An independent CPU audit checked the archived row identities and file hashes and recomputed the frozen analyses;
raw logits were not saved. Over the 32 metric prompts, the ratio of sequential-versus-batched \sysname{} KL to the
full-cache numerical reference is 1.034 with a prompt-bootstrap interval [0.932, 1.156]. The corresponding ratio
to the deepest-plane full-cache discrepancy is 0.628 [0.342, 1.195]; the latter interval includes one. These
ratios do not establish equivalence and the decoders follow their own trajectories. The fixed historical 22-prompt
subset is reported separately in the artifact. Other GPUs shared the host, so the timing records are observational
and do not establish isolated serving performance.

\textbf{L40S window replay.} A further Huginn replay combines the first 24 completed prompts of an interrupted run
with the remaining 24 prompts from another GPU on the same shared L40S VM. Source, model, runtime and prompt identity
were checked across the two archives. On 4,511 positions from the new full-cache trajectories, the deepest-plane arm
has seven KL$>0.1$ positions in two prompts and \sysname{} $W=64$ has two in two prompts. Full-cache trajectories match
the historical GH200 records for 28/48 prompts. The frozen 2,000-draw prompt bootstrap has 87 undefined $0/0$ draws
and 159 with a zero denominator and positive numerator; its upper interval endpoint is not finite. These results are
teacher-forced diagnostics, without archived raw logits or an isolated timing comparison.

\textbf{RTX PRO window replay.} A separate 48-prompt Huginn replay used two RTX PRO 6000 GPUs across an interrupted
and resumed rental; the original partial archive and both shards are retained and independently checked. On 4,649
positions from the newly generated full-cache trajectories, the deepest-plane arm has seven KL$>0.1$ positions in
two prompts, while \sysname{} $W=64$ has none. Full-cache trajectories match the historical GH200 records for 32/48
prompts. In the frozen 2,000-draw prompt bootstrap, 288 draws have an undefined $0/0$ tail ratio and are omitted;
the resulting [0,0] interval does not establish certainty. The raw logits were not saved, and the independent
trajectories preclude treating this as a common-context cross-device comparison or an isolated speed measurement.

The released evidence subset contains original per-position records for batch 8's 48-prompt window replay and batch R's
22 completed metric prompts and 12 timing prompts. A CPU-only script, \texttt{scripts/reproduce\_evidence.py}, verifies
SHA-256 digests and recomputes their summaries with the frozen analysis and prompt-level bootstrap seeds. This makes
these two result groups independently recomputable from recorded outputs; it is not an independent GPU reproduction.
Pre-registration documents, other batches' raw outputs and the raw GSM8K source used by the prompt builder remain
outside the release. The latter is unnecessary for executing the shipped token IDs, but is needed to reconstruct the
entire upstream input-building process. Hash agreement alone does not establish when the files were frozen.

The supplemental \texttt{scripts/collect\_repro\_env.py} records package versions and file hashes without running a
model. It can check supplied batch-R inputs and model files against the historical manifests. Its report is an
inventory, not a GPU test; \texttt{docs/H200\_REPRODUCTION.md} separates startup, numerical and timing checks.

\textbf{Post hoc replay audit.} The CPU recomputation script also reports per-prompt tail counts, domain counts, paired incidence changes
and leave-one-prompt-out reductions from the original batch-8 records. Its diagnostic values are generated separately
from the historical numeric macros. Of the original \ReplayBootstrapDraws{} paired prompt-bootstrap draws (seed 3),
\ReplayBootstrapUndefined{} have no tail event in either arm and thus an undefined $0/0$ ratio. The historical analysis
omits these NaN draws; the reported interval [\BeightBoneLo, \BeightBoneHi] uses the remaining
\ReplayBootstrapDefined{} draws. It describes the sampled tail-count ratio, not the fraction of affected prompts.

\subsection{Physical-cache reference implementation}\label{app:compact}

The supplemental \texttt{compact\_ripe.py}{} implements persistent storage separately from the frozen measurement code.
It keeps one full-length plane for each prelude, core and coda layer, and a ring for each non-final recurrence of each
core layer. With reserved capacity $C$ and maximum pass length $B=\gamma+1$, the allocated tensor bytes are
\begin{equation}
M_{\mathrm{alloc}}=b\left[(n_{\mathrm{pre}}+\ncore+n_{\mathrm{coda}})C+
\ncore(R-1)\min(C,W+B)\right].
\end{equation}
Here $b$ counts both K and V for one token and layer. The additional $B$ slots retain the current verification block
and any rejected suffix; Eq.~\eqref{eq:ripemem} instead counts the committed-prefix policy. A pass beginning at $t$
needs positions $t-W$ through $t+B-1$, a span of $W+B$. After accepting $k$ drafts, the next block begins at $t+k+1$,
so its oldest required recent position advances and does not require restoring an evicted prefix. During prefill,
attention consumes the current layer's full K/V tensors directly and retains only their tail for non-final loops.

CPU tests compare actual reads and masks with the frozen full-plane oracle, cover every acceptance length, rejected
suffix overwrites, ring wraparound, and tiny-model generation. The implementation still gathers dense attention inputs;
the equation excludes their temporary storage, model activations and allocator overhead. It supports one request and
has no fused serving kernel. These checks establish implementation behavior on CPU, not GPU performance.

The separate executable protocol fixes 12 cases at $W\in\{16,64\}$: four original prompt sequences and two repeated-token
contexts of lengths 1024 and 4096, each generating 64 tokens without early stopping. It compares compact and emulated
\sysname{} rows before three rotated timing trials against sequential full-cache decoding and speculative deepest-plane
reuse. Prefill, decode, persistent bytes and CUDA peaks are reported separately. Synthetic contexts test storage and
runtime only. An L40S run completed 10 of the 12 cases and all 120 associated timing records; all ten recorded
physical/emulated numerical checks passed. The first 4096-token-prefix case then raised an asynchronous CUDA
index-out-of-bounds assertion, leaving both 4096 cases incomplete. The original qualification therefore did not pass.
Huginn's configured RoPE table covers 4096 positions, whereas the first decode after a 4096-token prompt uses position
4096. A CPU check of the original table constructor confirms that this index is out of bounds; the asynchronous CUDA
trace does not independently locate the offending operation. A separate post hoc run extended that table with the
same formula and verified that its original entries were bitwise unchanged. For the two 4096-prefix cases, both
recorded compact/emulated checks passed and all 24 timing records were present. The logged persistent KV storage was
0.759 and 0.885\,GB for compact \sysname{} at $W=16$ and $W=64$, versus 11.614\,GB for emulated \sysname{}. These cases use a
single synthetic prompt beyond the checkpoint's configured context; raw log-probability arrays were not retained, and
they do not repair the original qualification or establish output quality or serving gains. This protocol supplies
neither task accuracy nor matched-memory comparisons with other methods.

For a common-model diagnostic with FlashLoop, a supplemental adapter connects the cache to pinned official
Ouro classes and changes both loop depth and output-head selection; selecting a head alone still executes all loops.
Tiny-model CPU tests cover native full-cache agreement, physical/emulated rows and shallow drafts. A separate NumPy
auditor recomputes numerical gates from saved output positions 2--32 and checks acceptance-schedule consistency;
it cannot reconstruct the first-token or full draft-chain decisions. Separate L40S runs on pretrained Ouro are reported
below; they do not constitute a matched-memory task-quality or serving comparison.

\subsection{Ouro L40S qualification and numerical failure}\label{app:ouro-l40s}

A separate post-submission qualification uses the pinned Ouro-1.4B checkpoint and its unchanged model classes on one
L40S (48\,GB), with eager bf16 attention. Other tasks used separate GPUs on the same VM; the Triton cache was not isolated.
Four fixed sequences comprise two diagnostic prompts and repetitions to 1024
and 4096 tokens. We use four loops, draft depth two, $\gamma=4$, $W\in\{4,16\}$, both draft policies and exactly 32
outputs. The fixed cohort has four native-prefill and 16 \sysname{} records; it does not measure task quality or performance.
All four native-versus-adapter prefill comparisons satisfy the predeclared 0.03 absolute tolerance and argmax gate.
Physical and emulated \sysname{} log-probability rows are bitwise equal in all 16 cases, with matching recorded schedules.
However, all 16 sequential-versus-batched comparisons fail the 0.03 log-probability gate: their maximum absolute
errors range from 0.313 to 4.410, and two cases also have argmax and target-trajectory disagreements. A post hoc
analysis of the 496 recorded positions finds three argmax disagreements across those two cases; the largest per-position
$\KL(\text{sequential}\|\text{batched})$ is 0.505. These dependent, teacher-forced positions do not measure task accuracy.
A separate NumPy audit of the archived arrays checks all 20 records and their hashes and reproduces the complete numerical failure. The
audit covers output positions 2--32; it cannot independently reconstruct the first-token or full draft-chain decisions.

A post hoc diagnostic isolates the first query after each of the two short prompts, comparing one-token and five-token
calls after identical prefills, with all later block tokens masked from that query. On the same L40S the maximum
log-probability difference is 0.227 and 0.125 in bf16, versus $1.34\times10^{-5}$ and $1.14\times10^{-5}$ in fp32;
the first-query argmax agrees in both precisions. A second diagnostic teacher-forces the two recorded arithmetic
trajectories through single-token and five-token blocks without draft chains. Each has one bf16 argmax disagreement,
versus none in fp32; maximum log-probability differences are 0.733 and 1.331 in bf16, and $2.63\times10^{-4}$ and
$6.48\times10^{-4}$ in fp32. These checks are consistent with sensitivity to execution shape and precision on these
inputs, but do not establish the cause of every later discrepancy or an application-level quality effect.
We also repeated the fixed 20-case cohort on one RTX PRO 6000. In bf16, all 16 sequential-versus-batched rows fail
the same 0.03 gate, with maximum errors between 0.373 and 2.087; physical and emulated rows match bitwise. On the
same GPU, a post hoc fp32 rerun with TF32 disabled passes all 16 comparisons (maximum error $2.98\times10^{-4}$).
Another shared L40S fp32 run passes all 16 (maximum $6.49\times10^{-4}$). All four native-prefill checks pass in
each of these runs. Independent audits recompute the saved arrays for output positions 2--32; the fp32 plan was
chosen after the bf16 failure and does not change its result or establish task quality.
The failed fixed cohort remains a failure; we do not relax its gates or infer a serving advantage from these results.

\section{Prediction ledger}\label{app:ledger}

Table~\ref{tab:ledger} lists every pre-registered prediction and its outcome, including the ones that failed. Batch R's
metric outcomes use its 22 completed prompts. Its prediction of an absolute trajectory count out of 32 remains
incomplete: 19 matches out of 22 cannot decide whether the final count would lie in 14--26 out of 32.

\begin{table}[H]\centering\small
\caption{Pre-registered predictions and outcomes (all reported, including failures).}\label{tab:ledger}
\begin{tabular}{@{}lp{8.4cm}p{5.6cm}@{}}\toprule
Batch & Prediction & Outcome\\\midrule
1 & coverage rises with $r$ & held\\
1 & head-median coverage at $r=4$ in 0.6--0.85 & GSM8K above (\BoneGsmMed), AlpacaEval held\\
1 & shallow top-$k$ above the static set by a margin & direction held; magnitude too small (``trivial'' on GSM8K)\\
2 & deep KV matters at matched size & held ($e^{\BtwoDeepBf}$ bf16, $e^{\BtwoDeepFp}$ fp32)\\
2 & \texttt{compress-s4} first-position agreement $\ge 0.95$ & held (\BtwoSfourFirst)\\
3 & deepest-plane cache indistinguishable from a precision change & failed (``exceeds'', via the tail)\\
3 & deepest-plane cache raises acceptance & held (\BthreeFullFour$\to$\BthreeSfourFour{} at $d=4$)\\
4 & reference asymmetry changes the ratio by $>10\%$ & failed (0.5--5\%)\\
4 & confirmatory set ``intermediate'' & failed (``exceeds'', one prompt)\\
4 & bf16 reruns bitwise identical & failed (\BfourRerunSame{} of positions)\\
5 & speed-up 2--3$\times$ unshared, 3.5--5$\times$ shared & held (\BfiveSpeedZero$\times$, \BfiveSpeedOne$\times$)\\
5 & WaveFront token-lossless (match $\ge 0.99$) & failed (\BfiveSameZero/40 texts identical)\\
6 & bf16 rerun noise is real & held\\
6 & deterministic mode removes it & failed\\
6 & WaveFront still not lossless in deterministic mode & failed by rule (8/8; chance $\approx$13\%)\\
6 & speed-up falls with batch size & held\\
6 & shared $\ge 1.5\times$ best lossless at 95 GB & held (\BsixRatio$\times$)\\
6 & unshared runs out of memory between batch 24 and 32 & failed (32 fits)\\
6b & source is cuDNN attention & held\\
7 & reruns reproducible with cuDNN off & held\\
7 & ratios rise to 2.5--4 with cuDNN off & failed (unchanged)\\
7 & WaveFront 40/40 with cuDNN off & failed (\BsevenWsameZero/40, \BsevenWsameOne/40)\\
8 & batched deepest-plane verification $\ne$ sequential (ratio 3--30) & failed (\BeightAoneRatio, numerics)\\
8 & batched verification closer to the full cache & failed (\BeightAtwoRatio)\\
8 & tokens per verification 3.0--3.6 (full), $+\ge 0.5$ (shared) & held (\BeightTpvFull, \BeightTpvSfour)\\
8 & speed-up 1.5--2.3$\times$ (full), 2.0--3.0$\times$ (shared) & held (\BeightSpeedFull, \BeightSpeedSfour)\\
8 & window 64 removes most tail positions & held (\BeightBoneRed)\\
8 & two planes remove 20--50\% & failed (none)\\
8 & per-head $\approx$ all top-$k$ & failed (\BeightCtwoRatio$\times$ worse)\\
8 & filtering costs $>3\times$ sharing; costs add & held (\BeightConeRatio, \BeightCthreeRatio)\\
8 & unshared OOM at batch 38 & failed (fits)\\
8 & shared/lossless ratio 2.4--2.8 & failed (\BeightDratio)\\
9 & Raven: deep KV matters, deep coefficient $\ge 3$ (bf16) & held (\BnineDeepBf)\\
9 & Raven: deepest-plane first-position agreement $\ge 0.95$ & held (\BnineSsixFirst)\\
9 & Raven: deepest plane ``exceeds'', more than on Huginn & failed (indistinguishable, no tail)\\
R & batched \sysname{} as numerics, ratio 0.7--1.3 & held (\BrOneRatio)\\
R & identical trajectories as in batch 8's full cache (14--26 of 32) & incomplete (\BrDdEqSeq/\BrPrompts; 32 planned)\\
R & deepest-plane drafts: 4.0--4.6 tokens/verification, keep $\ge70\%$ of the gain & held (\BrTauDd, \BrRecDd)\\
R & drafts following $\pi_{\textsc{\sysname}}$ keep $\le50\%$ of the gain & held (\BrRecRipe)\\
R & \sysname{} target closer to the full cache than the deepest plane & failed (no difference, \BrThreeRatio)\\
R & fewer positions with KL$>0.1$ than the deepest plane & held (\BrTailDd{} vs.\ \BrTailSfour)\\
R & speed-up 2.2--2.8$\times$, not below the full cache & held (\BrSpeedDd$\times$)\\
R & sequential \sysname{} at least 0.9 of sequential full-cache speed & failed (\BrSeqRatio)\\
R & cost model within 10\% & held (\BrErrRipe, \BrErrDd)\\\bottomrule
\end{tabular}
\end{table}

\section{Fitted parameters}\label{app:fits}

The fitted parameters and historical tables in this section are generated by one script from the batch summaries and
raw per-position records (Tables~\ref{tab:thr}, \ref{tab:accept}; Figure~\ref{fig:ripenv}). The later FlashLoop
diagnostic uses a separate independent audit of its archived arrays.

\begin{table}[H]\centering\small
\caption{Ripening fits (keys; values in parentheses). $\rho$ and $c$ from fp32, loops 8--31; $\Delta_\infty$ fitted to bf16
with $\rho, c$ fixed.}
\begin{tabular}{@{}lcc@{}}\toprule
 & Huginn & Raven\\\midrule
$\rho$ & \MRhoHugK{} (\MRhoHugV) & \MRhoRavK{} (\MRhoRavV)\\
half-life (loops) & \MHalfHugK{} (\MHalfHugV) & \MHalfRavK{} (\MHalfRavV)\\
$\Delta_\infty$ (bf16) & \MFloorHugK{} (\MFloorHugV) & \MFloorRavK{} (\MFloorRavV)\\
deep-loop $\Delta$, bf16 & \MDeepBfHugK{} (\MDeepBfHugV) & \MDeepBfRavK{} (\MDeepBfRavV)\\
deep-loop $\Delta$, fp32 & \MDeepFpHugK{} (\MDeepFpHugV) & \(\MDeepFpRavK\) (\(\MDeepFpRavV\))\\
max $|\log$ residual$|$, fp32 & \MResFpHugK & \MResFpRavK\\\bottomrule
\end{tabular}
\end{table}

\begin{table}[H]\centering\small
\caption{Throughput fits $T(B)=B/(a+bB)$ per configuration of WaveFront's batched emulator (Section~\ref{sec:memory}):
saturation throughput $T_\infty=1/b$ (tokens/s), half-saturation batch $B_{1/2}=a/b$ and largest relative residual. Batch-8
series cover a narrow batch range and ran on another node; their parameters are descriptive only.}\label{tab:thr}
\begin{tabular}{@{}cclcrrrrr@{}}\toprule
Batch & Slots & Scheduler & $\alpha$ & $B$ range & Cells & $T_\infty$ & $B_{1/2}$ & Residual\\\midrule
6 & 0 & AR & --- & 1--32 & 5 & 394 & 69 & 4\%\\
6 & 0 & DtV & 0.878 & 1--32 & 5 & 449 & 33 & 2\%\\
6 & 0 & DtV & 0.965 & 1--32 & 5 & 612 & 34 & 3\%\\
6 & 0 & WaveFront & 0.878 & 1--32 & 5 & 244 & 13 & 3\%\\
6 & 0 & WaveFront & 0.965 & 1--32 & 5 & 310 & 11 & 6\%\\
6 & 1 & AR & --- & 1--64 & 6 & 461 & 83 & 4\%\\
6 & 1 & DtV & 0.878 & 1--64 & 6 & 494 & 38 & 3\%\\
6 & 1 & DtV & 0.965 & 1--64 & 6 & 656 & 39 & 6\%\\
6 & 1 & WaveFront & 0.878 & 1--64 & 6 & 535 & 30 & 4\%\\
6 & 1 & WaveFront & 0.965 & 1--64 & 6 & 689 & 27 & 6\%\\
6 & 4 & AR & --- & 1--64 & 6 & 395 & 69 & 6\%\\
6 & 4 & DtV & 0.878 & 1--64 & 6 & 473 & 35 & 8\%\\
6 & 4 & DtV & 0.965 & 1--64 & 6 & 539 & 29 & 10\%\\
6 & 4 & WaveFront & 0.878 & 1--64 & 6 & 521 & 29 & 9\%\\
6 & 4 & WaveFront & 0.965 & 1--64 & 6 & 644 & 24 & 8\%\\
8 & 0 & AR & --- & 32--38 & 4 & 372 & 56 & 2\%\\
8 & 0 & DtV & 0.878 & 32--38 & 4 & 380 & 18 & 2\%\\
8 & 0 & DtV & 0.965 & 32--38 & 4 & 503 & 18 & 2\%\\
8 & 0 & WaveFront & 0.878 & 32--38 & 4 & 219 & 8 & 0\%\\
8 & 0 & WaveFront & 0.965 & 32--38 & 4 & 280 & 7 & 0\%\\
8 & 1 & AR & --- & 64--256 & 5 & 308 & 28 & 2\%\\
8 & 1 & DtV & 0.878 & 64--256 & 5 & 443 & 24 & 1\%\\
8 & 1 & DtV & 0.965 & 64--256 & 5 & 555 & 20 & 2\%\\
8 & 1 & WaveFront & 0.878 & 64--256 & 5 & 483 & 18 & 2\%\\
8 & 1 & WaveFront & 0.965 & 64--256 & 5 & 605 & 14 & 2\%\\
\bottomrule

\end{tabular}
\end{table}

\begin{table}[H]\centering\small
\caption{Acceptance model check on the batch-3 simulation ($\gamma=4$): mean conditional acceptance $\alpha_d$, its spread
across the four draft positions, number of blocks, and tokens per verification predicted by Eq.~\eqref{eq:tau} vs.\
simulated. Rejection $1-\alpha_d\approx k\rho_a^d$ fits $k=\MRejKFull$, $\rho_a=\MRejRhoFull$ (full cache) and
$k=\MRejKDeep$, $\rho_a=\MRejRhoDeep$ (deepest plane).}\label{tab:accept}
\begin{tabular}{@{}lcccrcc@{}}\toprule
Cache & $d$ & $\alpha_d$ & Spread & Blocks & $\tau$ model & $\tau$ sim.\\\midrule
full cache & 2 & 0.637 & 0.041 & 2280 & 2.46 & 2.44\\
full cache & 4 & 0.805 & 0.062 & 1676 & 3.40 & 3.31\\
full cache & 8 & 0.932 & 0.021 & 1281 & 4.37 & 4.34\\
deepest plane & 2 & 0.800 & 0.041 & 1619 & 3.36 & 3.41\\
deepest plane & 4 & 0.928 & 0.008 & 1276 & 4.33 & 4.33\\
deepest plane & 8 & 0.984 & 0.005 & 1139 & 4.84 & 4.85\\
\bottomrule

\end{tabular}
\end{table}

\begin{figure}[H]
\centering
\includegraphics{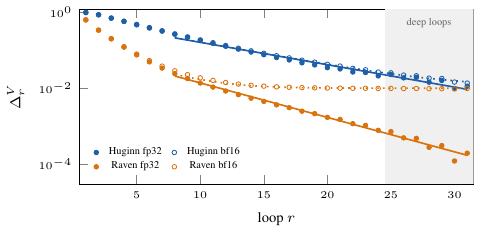}
\caption{Ripening of prefix values, as Figure~\ref{fig:mech}c shows for keys: fp32 (filled), bf16 (open), geometric
fit to fp32 over loops 8--31 (solid) and the same fit with one bf16 floor (dotted).}\label{fig:ripenv}
\end{figure}

\section{Post hoc common-model FlashLoop diagnostic}\label{app:flashloop-ouro}

We additionally ran pinned FlashLoop commit \texttt{aeaadee} with its default KIVI reader on a separate L40S VM,
using the same Ouro-1.4B checkpoint and four inputs above. This run used bf16, four loops, PyTorch 2.14/CUDA 13 and
Transformers 4.56.2; other VMs shared the physical host. We generated 32 tokens without early stopping with sequential
full cache, sequential physical \sysname{} ($W=4,16,64$), and FlashLoop's engine. These are different implementations:
FlashLoop also uses token and attention sparsity and residual quantization. The full and \sysname{} runs use eager attention;
the FlashLoop engine uses its own attention and cache readers. The observed native/adapter full-prefill rows agree
bitwise on all four inputs. All 20 free-generation records and their generated log-probability arrays were archived.

We then forced each method along the same full-cache trajectory and independently recomputed KL divergence and
argmax differences from the archived arrays for output positions 2--32 (Table~\ref{tab:flashloop-ouro}). The two
natural prompts are diagnostics, and the long contexts repeat one of them; the 31 positions per case are dependent.
The $W=64$ \sysname{} results are close on the two short prompts but differ substantially on these synthetic long contexts.
FlashLoop is closer to full cache on the repeated contexts in this sample. At length 4096, recorded persistent KV
allocation is 0.825\,GB for \sysname{} $W=16$, 0.853\,GB for \sysname{} $W=64$, and 0.601\,GB for FlashLoop; the cache allocation
policies differ. These numbers are not matched budgets, quality scores, or serving measurements. In particular,
sequential \sysname{} here does not repair the speculative numerical failures above.

\begin{table}[H]\centering\small
\caption{Ouro common-trajectory diagnostic: mean $\KL(\text{full}\|\text{arm})$ / argmax disagreements out of 31
dependent positions, for one 32-token full-cache trajectory per input.}\label{tab:flashloop-ouro}
\begin{tabular}{@{}lrrrr@{}}\toprule
Input & \sysname{} $W=4$ & \sysname{} $W=16$ & \sysname{} $W=64$ & FlashLoop\\\midrule
Arithmetic & 0.413/4 & 0.0142/1 & 0.000100/0 & 0.0216/1\\
Explanation & 0.795/7 & 0.00915/1 & 0.000320/0 & 0.0167/1\\
Repeated 1024 & 2.70/13 & 1.80/11 & 1.17/7 & 0.0260/1\\
Repeated 4096 & 4.31/23 & 2.31/14 & 2.51/12 & 0.00197/1\\\bottomrule
\end{tabular}
\end{table}

\textbf{Post hoc task-format diagnostic.} We also re-encoded the 24 frozen GSM8K questions with the Ouro chat
template and generated at most 160 tokens per question with full cache, \sysname{} $W=16,64$, and FlashLoop's default
KIVI reader. These questions are a reused convenience sample, not a held-out benchmark. Using the original strict
\texttt{answer is <number>} extractor, full cache, \sysname{} $W=16$, \sysname{} $W=64$, and FlashLoop produced parseable answers
on 2, 1, 2, and 0 of 24 questions; respectively 0, 1, 1, and 0 extracted answer strings exactly matched the
reference. Exact string comparison counts \texttt{6.00} and \texttt{6} as different. The output
limit was reached for 20, 24, 21, and 24 questions, respectively. The small number of parsed answers and frequent
truncation preclude an accuracy ranking; an unparseable output is not established to be incorrect. The independent
audit preserves the generated token sequences, stopping reasons, and numerical preflight checks. The methods use
different cache allocations and attention implementations, so these results are not matched-memory comparisons.

After observing this truncation, we repeated the same questions with a 512-token cap on L40S and, separately, on
RTX PRO 6000. In the same arm order, EOS counts were 11, 0, 2 and 14 out of 24 on L40S, and 10, 0, 1 and 12 on
RTX PRO. Length-cap counts were 13, 24, 20 and 10 on L40S, and 14, 23, 21 and 12 on RTX PRO; the remaining
outputs stopped at a next-question marker. These are post hoc, reused samples, and the frequent incomplete outputs
still prevent interpreting the strict extractor as a reliable task-accuracy comparison.

We then reran the same 24 questions with a 2048-token cap, stopping only at EOS, with all four methods for each
question on one GPU (12 questions on two L40S, 12 on two RTX PRO 6000; within-GPU pairing only). Full cache and
FlashLoop reached EOS on 23 of 24 questions, \sysname{} $W=16$ on none and $W=64$ on 2; one of these is a 38-token answer identical to the
full cache's. Of the length-capped \sysname{} outputs, 23 ($W=16$) and 19 ($W=64$) end in
near-verbatim repetition (eight-gram diversity below 0.05 over the last 512 tokens), against none for the other methods.
Restricted to EOS outputs, a post hoc last-number extraction matched the reference on 21 of 23 (full), 23 of 23
(FlashLoop) and 2 of 2 ($W=64$); this conditional score is not accuracy over the sample. A teacher-forced replay of the
first question's capped $W=16$ and $W=64$ outputs finds the physical cache bitwise identical to the emulated policy at
all 2048 positions, so the repetition is not a divergence between physical and emulated storage; because both paths share
the substitution code, the replay cannot exclude an error common to both. The chat-template prompts have 318--415 tokens, so
with $W\le 64$ most prompt positions are read at the deepest plane. This agrees with Ouro's own report that sharing the last
step's KV across the prompt costs more than ten GSM8K points, while decode-only sharing is nearly lossless \citep{ouro}. On
Ouro and this reused sample, a recent window of 16 or 64 positions with deepest-plane K/V elsewhere does not preserve the
model's ability to finish an answer; a policy that keeps all planes of the prompt is not evaluated here.

\end{document}